\documentclass{statsoc}

\usepackage[a4paper]{geometry}
\usepackage{graphicx}
\usepackage[textwidth=8em,textsize=small]{todonotes}
\usepackage{amsmath}
\usepackage{natbib}
\usepackage{soul}

\usepackage[colorlinks,hypertexnames=false]{hyperref}
\usepackage{amsfonts}
\usepackage{calc}
\usepackage{marginnote}

\newtheorem{theorem}{Theorem}
\newtheorem{assumption}{Assumption}
\newtheorem{lemma}{Lemma}

\newcommand \mcP{\mathcal{P}}

\newcommand \mbR{\mathbb{R}}

\newcommand \bomega{\boldsymbol{\omega}}

\DeclareMathOperator{\E}{E}

\DeclareMathOperator{\Var}{Var}
\usepackage{listings}

\newcommand{\note}[1]{\textcolor{black}{#1}}

\numberwithin{equation}{section}

\title[Prior Informativeness]{Prior sample size extensions  
for assessing prior \\
\note{impact} and prior--likelihood discordance}

\author{Matthew Reimherr}\coaddress{Matthew Reimherr, Department of Statistics, The Pennsylvania State University, 411 Thomas Building, University Park, PA 16802, USA. \textit{mreimherr@psu.edu}}
\address{The Pennsylvania State University, University Park, PA, USA.}
\author{Xiao-Li Meng}
\address{Harvard University, Cambridge, MA, USA.}
\author[Reimherr, Meng, and Nicolae]{Dan Nicolae}
\address{The University of Chicago, Chicago, IL, USA.}
\begin{document}
\kern -.2in
\begin{abstract}
This paper outlines a framework for quantifying the prior's  contribution to
posterior inference in the presence of prior-likelihood discordance, a broader  concept than the usual notion of prior-likelihood conflict.
We achieve this dual purpose by  extending the classic notion of \textit{prior sample size}, $M$, in three directions: (I)  estimating $M$ beyond  conjugate families;  (II) formulating $M$ as a relative notion, i.e., as  a function of the likelihood sample size $k, M(k),$ which also leads naturally to a graphical diagnosis; and (III) permitting  negative  $M$, as a measure of prior-likelihood conflict, i.e., harmful discordance. Our  asymptotic regime  permits  the prior sample  size to grow with the likelihood data size, hence making  asymptotic arguments meaningful for investigating the impact of the prior relative to that of likelihood. It leads to a simple asymptotic formula for quantifying the impact of a proper prior that only involves computing a centrality and a spread measure of the prior and the posterior. We use simulated and real data to illustrate the potential of the proposed framework, including quantifying how weak is a ``weakly informative" prior  adopted  in  a study of  lupus nephritis. \note{Whereas we take a pragmatic perspective in assessing the impact of a prior on a given inference problem under a specific evaluative metric, we also touch upon conceptual and theoretical issues such as using improper priors and permitting priors with asymptotically non-vanishing influence.}
\end{abstract}

\section{Motivation and Illustration}\label{sec:moti} 

The difficulty in choosing priors and fully understanding their impact on statistical  analyses has been a primary concern of Bayesian methods since their inception. 
The common approach to alleviate such concerns is to conduct a  sensitivity analysis,  investigating how the results are affected by perturbations of the prior. However, such an approach does  not typically reveal how a chosen prior has actually contributed to the analysis in comparison to the information from the data, as captured by the  posited likelihood.  Although many have asked questions along the line  ``How much of your conclusion is actually due to your prior assumptions?",  to the best of our knowledge, there are no well-recognized approaches to quantitatively address such legitimate inquires.     

As the  adoption of Bayesian tools continues to grow, we need quantitative  assessments on the impact of priors, permitting at least a check on their impact compared to the likelihood.   
This  is desirable, both scientifically and statistically.  A posterior inference with 45\% prior information contribution may affect our decisions rather differently, at least psychologically, from one with only 5\% prior contribution. %
However, quantifying the impact of a prior has been a very challenging task, partially explaining the lack of routinely adopted methods. 
  A key difficulty is that the information from the prior may  be in conflict with that from the likelihood to a point that  it can actually ``subtract" rather than ``add" to an analysis.  Recently \citet{efron:2015} explored frequentist properties of Bayesian estimates, illustrating that, in many ways, researchers are often still unable to understand/quantify the impact of their priors on their inference.  Our paper aims  to make a substantive contribution in this direction, though we by no means declare that we have found \textit{the} solution as this is an area that requires much more research.
  
\subsection{Building Upon the Classic Notion of  Prior Sample Size}\label{sec:size}
This paper presents a strategy for simultaneously assessing the degree of conflict between the prior and likelihood, and quantifying the information contribution of the prior to the posterior. 
We accomplish this by extending the easily interpretable information metric \textit{prior sample size} (PSS), a common notion in the literature of conjugate priors \citep[e.g.,][]{diaconis1979conjugate,gutierrez1997exponential,meng2002single}.   As is well known, a conjugate  prior can be equated to the posterior from a prior study with, say, $M$ i.i.d. (hypothetical) observations and a 
baseline  ``non-informative" prior.   This equivalence provides a  concrete practical guideline. If our likelihood is based on $N$ i.i.d.
observations from the same conjugate family, then we  can consider that  the conjugate prior has contributed $M/(N+M)\times 100\%$ of the information to our posterior inference.

Given its practical appeal,  multiple  efforts have been made to extend the concept of PSS beyond conjugate families.  The approach  by \cite{clarke1996implications}
is particularly significant, taking advantage of reference priors \citep{bernardo1979reference}, which are equivalent to  Jeffreys priors in univariate cases \citep[][]{george1993obtaining}.  Specifically, by minimizing their relative entropy (i.e., Kullback-Leibler divergence), Clarke's approach matches a target prior---typically considered to be  informative---with the posterior based on a likelihood from a given family and a given reference prior.   The resulting likelihood function is then interpreted as representing implicit data information in the target prior, relative to the reference prior. The PSS of the target  prior is then approximated by this  likelihood data size, though \cite{clarke1996implications} also identifies the actual data values used by the likelihood (typically not unique).   Subsequently,  \cite{clarke:yuan:2006} developed closed-form expressions for the prior sample size, and \cite{lin:pittman:clarke:2007} examined how to quantify the information content with non i.i.d. data.  \cite{ginebra:2007} discussed, more generally, how to quantify the information content in an experiment, and \cite{berger:bayarri:pericchi:2014} investigated how to quantify effective sample sizes of various linear models.  Recently \cite{wiesenfarth2019quantification} explored the usage of historical data for quantifying PSS in the context of clinical trials.

In addition, \cite{morita:thall:muller:2008} took a similar approach but with a different baseline and divergence. Their  baseline prior was constructed by keeping the prior mean and correlations (for multivariate cases) the same as the target prior, but with the prior variance greatly inflated to render $\epsilon$-amount of information.  Their  divergence is based on the trace of the (expected) curvature of the log density, which they reported was best after extensive trial and error. This is intuitive as the curvature of the log density is to the prior variance what Fisher information is to the posterior variance.  This pairing  between the  divergence  and the baseline prior  was also emphasized by \cite{clarke1996implications},  because the reference prior is  the minimizer of an expected K-L divergence, and hence his pairing also guarantees his  PSS to be  non-negative when it exists.  However,  permitting a negative PSS turns out to be the key to resolve the thorny issue of prior-likelihood conflict when measuring the impact of prior information, as we will discuss in Section~2.  The methods of  \cite{morita:thall:muller:2008} were further illustrated on biomedical applications in  \cite{morita:thall:muller:2010}, and extended to hierarchical models in \cite{morita:thall:muller:2012}.

A commonality  between the settings in \cite{clarke1996implications} and \cite{morita:thall:muller:2008} (and their subsequent extensions) is that both treat the likelihood model as a device 
for measuring  the information in the target prior,  with the hypothetical observations  optimized over  or averaged over, respectively. Hence these are \textit{pre-data} measures, most useful for design purposes and theoretical investigations. In this paper, however, we address  a harder and more common \textit{post-data} question:  how many observations are required, approximately, to match the prior's contribution---in terms of some statistical efficiency---to  a posterior inference based on a particular \textit{likelihood function} from a set of observed data?
That is, we intend our PSS to be \textit{inherently  data dependent}. 

In theory, we all hope that, at the very least, our prior does no harm. However,  in practice, typically there are some degrees of prior-likelihood conflict. This can lead to, for example, a $95\%$ posterior interval that is wider than an analogous  $95\%$ confidence interval, or which has deficient coverage.  This can occur regardless of the correctness of the prior or likelihood, because a particular data set from a known model
can still exhibit ``tail" behavior.  Any measure of the prior contribution to a particular posterior inference must  then allow for the possibility of negative prior contribution. But  how should we formulate negative ``sample size"? We report  a practical way to circumvent this problem by matching two posteriors corresponding to two priors (target and baseline), instead of the prior-posterior matching as in \cite{clarke1996implications} or  in    \cite{morita:thall:muller:2008}.  

The more fundamental  question is the meaning of measuring  the statistical efficiency in a particular study. All statistical inference paradigms  require a specification of \textit{reference replications} \citep[see, e.g.,][]{liumeng2016}, because otherwise there is no variation, and hence no information, to speak of.  The reference replications in \cite{clarke1996implications} and \cite{morita:thall:muller:2008} are pre-data hypothetical observations as respectively specified, designed as frequentist measures, averaging over many hypothetic data sets that we will never observe.
But if we insist on using  the Bayesian replication, that is, all data sets that are exactly the same as the observed ones, then our entire information measure will be driven by the prior, the only source of variation.  To avoid either extreme \citep[see][for reasons for this avoidance]{liumeng2016}, we adopt a compromise: we measure information/efficiency with respect to all data sets that are \textit{exchangeable} with the observed data, that is, they are not identical  to the observed data but they share the same generating model with the same parameter value for the latter. We do not know this parameter value, but when we have adequate internal replications (e.g., with   i.i.d. data), we will be able to estimate our measure via common methods such as bootstrap.  

Before we proceed to illustrate the key ingredients of our proposal, we emphasize  that the need for choosing a baseline prior, as above  and in other similar works such as \cite{evans:jang:2011},  is  unavoidable because it is mathematically impossible  to represent  ignorance via a probability distribution \cite[see][Proposition 2.1]{martin2016inferential}. In responding to an insightful question raised by a reviewer, we will also report a ``prior size paradox" caused by this impossibility (see Section~\ref{s:para}).   Our preference therefore is to choose a prior that represents what a practitioner would adopt without real prior information, such as those documented in  \cite{kass:wasserman:1996}.    We also note that various developments on deviance information criteria  such as \cite{spiegelhalter:etal:2002}, \cite{watanabe:2010} and \cite{watanabe:2013}, and others as reviewed in \cite{gelman:etal:2014}, are similar in spirit to our goal of quantifying prior-likelihood conflicts. They  also  provide  information deviances to be examined  within our framework,  in addition to the quadratic loss measures, a focus of this paper.  Furthermore,  the concept of \textit{surprise} has been used in \cite{evans:1997}, \cite{evans:moshonov:2006}, \cite{bousquet:2008}, and \cite{evans:jang:2011} for a variety of procedures including detecting prior-likelihood conflicts. In particular,  \cite{evans:moshonov:2006} provided methods to check for conflict in proper priors, an assumption we avoid making due to the heavy reliance on improper priors in practice, whereas \cite{bousquet:2008} expanded on this work, giving a binary decision rule for determining if there is conflict.
\vskip -0.1in
\subsection{The Normal Enlightenment}\label{sec:norm}
As usual,  a normal distribution example sheds much light on what lies ahead, and permits exact analytic results.  Our main theoretical results, given in Section~2,  show that under an asymptotic regime that  permits the influence of the prior to grow with the likelihood data size, the exact normal results  are special cases of the asymptotic results for a large class of likelihood-prior models. We will also use this example (and others) in Section~4 to check the implementation and computational procedures  outlined in Section 3.

We begin by assuming $\vec X_n=\{X_1,\dots,X_n\}$ to be an i.i.d. sample from $N(\mu, \sigma^2_0)$, where for simplicity of illustration, we assume $\sigma^2_0$ is known. We adopt the usual conjugate prior on $\mu$, $\pi = N(\mu_{\pi}, \sigma^2_0/m)$,  but we parameterize the prior variance as $\sigma^2_0/m$ since if our prior was set according to a previous data set also from the model $N(\mu, \sigma^2_0)$,
say $\{Y_1, \ldots, Y_m\}$,  then the variance of the previous MLE of $\mu$ would be $\sigma^2_0/m$.  For the baseline prior, $\pi_b$, we take $m=0$, i.e., a constant prior.  With a slight abuse of notation, we 
let $\vec X_k$ be a generic notation for any  subset of $\vec X_n$ with size $k \le n$, and $\bar X_k$ denotes the corresponding sample average.   The posterior of $\mu$ given   $\vec X_k$ under either prior is normal with, respectively, 
   \begin{align}
& \E_\pi[\mu |\vec{X}_k] = \frac{m \mu_\pi+ {k \bar{X}_k}}{m + k},
&& \Var_\pi [\mu |\vec{X}_k] =  \frac{\sigma^2_0}{m + k}; \label{e:norm}\\
& \E_{\pi_b}[\mu |\vec{X}_k] =\bar{X}_k,
&&\Var_{\pi_b} [\mu | \vec{X}_k] =   \frac{\sigma_0^2}{k}. \label{e:norb}
\end{align}

It is natural to ask, how $\pi$ has changed our posterior inference for $\mu$ compared to the baseline?  To be specific, let us examine the posterior mean-squared error  (MSE), averaged over all data sets that are generated (under the normal model) by the same $\mu_0$ that generated our observed data $\vec X_n$.  Therefore,  denoting  by $U$ this  expected MSE,  we would like to compare
\[
U_\pi(k) =  \frac{\sigma^2_0}{m + k} + \E \left[  \frac{ m \mu_\pi+ {k \bar{X}_k}}{m+ k} - \mu_0 \right]^2
\quad \mbox{ with } \quad
U_{\pi_b}(k) = \frac{\sigma^2_0}{k}  +  \E \left[   \bar{X}_k - \mu_0 \right]^2.
\]  
Let $\Delta =\sqrt{m}(\mu_0 - \mu_\pi)/\sigma_0,$  which encapsulates the degree of \textit{discordance} between the prior and the likelihood. It is then straight forward to show that
\begin{equation}\label{eq:norm}
U_\pi(k) = \frac{[2k + m (1+\Delta^2)]\sigma^2_0}{(m + k)^2}   
\quad \mbox{ and } \quad
U_{\pi_b}(k) = \frac{2\sigma^2_0}{ k}.
\end{equation}
Our approach is to  find $M(k)$ such that $U_{\pi_b}(k+M(k)) = U_\pi(k)$,
so $M(k)$ can be viewed as the PSS of $\pi$ relative to $\pi_b$. 

After some algebra, we can express  
\begin{equation}\label{eq:normM}
M(k) = m \left[1-   \frac{(r+1)(\Delta^2 - 1)}{  2(r+1) +r(\Delta^2 - 1)} \right], 
\end{equation}
where $r=m/k$ is the nominal prior size relative to the likelihood data size. Expression (\ref{eq:normM}) reveals something unexpected: $M(k)=m$, the \textit{perceived} PSS, if and only if $\Delta^2=1$.  This may surprise those who expect that $M(k)=m$ when $\Delta=0$. However, if our prior was specified according to a prior data set $\{Y_1, \ldots, Y_m\}$, then we would have set $\mu_\pi=\bar Y_m$, and hence $\Delta^2=m(\bar Y_m-\mu_0)^2/\sigma^2_0$, which is distributed as $\chi^2_1$ when the prior data set is indeed from the same population. That is, on average we should expect $\Delta^2$ to be 1, not 0. Therefore, $\Delta^2<1$ means we have a ``fortuitous" prior (as compared to a no-conflict prior), and hence $M(k)>m$ because of the additional ``lucky" information brought in by $\mu_\pi$. When $1<\Delta^2 < 1+ 2(1+r)$, we have $0<M(k)<m$, meaning that, although the prior is not as informative as its nominal size $m$ advertises, it is still helpful in the sense of reducing the MSE over using the baseline.  However, when $\Delta^2 \ge 1+ 2(1+r)$, \note{the prior has zero or negative impact},  because $M(k)\le 0$. 1

In summary, (\ref{eq:normM}) tells us that, \note{\textit{with respect to the impact on MSE for estimating }}$\mu$, 
\begin{equation}\label{e:key}
\begin{array}{lcr}
\mbox{When}\ \Delta^2<1, &  M(k) > m;  &  \mbox{lucky prior} \\
\mbox{When}\ \Delta^2=1, &  M(k) = m; &  \mbox{``advertised'' prior} \\ \mbox{When}\ 1<\Delta^2 < 1+ 2(1+r), & 0< M(k) < m; &  \mbox{unlucky but helpful prior} \\
\mbox{When}\ \Delta^2= 1+ 2(1+r), & M(k)=0; &  \mbox{\note{zero-impact} prior} \\
\mbox{When}\ \Delta^2 > 1+ 2(1+r), &  -k \le M(k) <0;  &  \mbox{harmful prior} \\
\end{array} 
\end{equation}

This calls for a more general concept of \textit{prior-likelihood discordance} than prior-likelihood conflict to describe the lack of harmony between our \textit{likelihood model} and \textit{prior model}, because not all such discordance is harmful, which the phrase ``conflict" would suggest.  Indeed, $\Delta^2=0$ represents our most fortunate case, with the prior mean being exactly the true parameter value and hence $M(k)$ reaches its maximal  value $m(3+2r)/(2+r)$. We observe that this maximal value is always between $1.5 m$ and $2m$. We suggest that the term \textit{prior-likelihood conflict} is reserved for cases when our prior becomes harmful, that is, when $M(k)<0$. \note{We must emphasize that we take a pragmatic perspective in suggesting these terms, by considering primarily the impact of the target prior on the chosen inference with respect to a specified evaluative metric (and a baseline prior). Hence a zero-impact prior does not mean a zero-information prior (which is a self-contradictory phrase in the Bayesian framework), nor does a helpful prior imply no harmful consequences, such as lack of robustness; see \citet{al2017optimal}. }

Regardless of the value of $\Delta^2$, we see from (\ref{eq:normM})
that $M(k)$ is a strictly decreasing function of $k$, unless $\Delta^2=1$
when it is a constant function. Its decreasing rate is controlled by $\Delta^2-1$, with the most rapid decreasing occurring when $\Delta^2=\infty$, in which case $M(k)=-k$. We see that,  whenever there  is a prior-likelihood discordance (regardless of being lucky or unlucky), the slope of $M(k)$ will be negative.
Its extreme value, $-1$, is also achieved if and only if $\Delta=\infty$, which means that the prior-likelihood conflict is so extreme that it wipes out the entire likelihood. When we treat $k$ as a continuous index, the derivative of $M(k)$ will be always bounded below by $-1$. 

Therefore, this  normal example leads to (at least) five observations: 
\begin{description}
\item(I) PSS is a relative concept, relative to the size of the likelihood sample size (LSS);
\item(II) The dependence of PSS on LSS is governed by the prior-likelihood discordance;
\item(III) The prior-likelihood discordance can be both beneficial and harmful;  
\item(IV) PSS can take negative values, when the prior-likelihood discordance is severe;
\item(V) The PSS as a function of LSS, $M(k), $ has a slope that is bounded below by $-1$. 
\end{description}
The main contribution of this paper is to show, theoretically and empirically, that these observations hold rather generally. Theoretically, we show that the normal formula (\ref{eq:normM}), not surprisingly, holds asymptotically for a rather general class of distributions and hence (\ref{e:key}) holds as well. This asymptotic approximation provides a quick (and not too dirty) assessment of the prior impact almost as a byproduct of the original posterior computation. But for those who are willing and able to do more, we also describe a finite-sample bootstrap-like method to estimate $M(k)$, especially its slope, as a function of $k$,  which provides a diagnostic tool for detecting the discordance.  

A reviewer's comment also reminded us to stress  that the  notion of \textit{prior-likelihood discordance} is a qualitative and absolute concept,  intended to indicate any kind of incompatibility   between the prior and likelihood (function), harmful or not.  In contrast,  the classic notion of \textit{prior sample size} (PSS) is a quantitative and relative concept, designed to provide a practically appealing measure to numerically index the strength or weakness of an adopted prior with respect to our likelihood function. Whereas both concepts are needed, quantitative measures can do more harm than qualitative ones because of their seductive nature of being precise, regardless of their validity. It is therefore critical for those of us who develop such measures to be explicit about their limitations and potential misuse, a practice we follow whenever appropriate.    

The remainder of the paper is organized as follows. Section \ref{s:meth} provides our general framework, and implements it asymptotically, as well as for finite-sample i.i.d. data.  Section \ref{s:th} establishes theoretical results for a large  class of distributions to justify the implementations in Section~\ref{s:meth}.  Section \ref{s:ex} gives a simulation study, and a real-data application.  Section \ref{s:disc} concludes with a discussion of complications, limitation and   open problems. Some secondary proofs and  technical verifications are in the  online supplemental material.  All computations were done using {\tt R} and the accompanying code can be found at the corresponding authors website.

\section{A General Formulation of Prior Sample Size} \label{s:meth}
Let $Y\in S$ represent a data set and $f(y|\theta)$ its density, with $\theta \in \Theta$ being the model parameter, and $\theta_0$  the  value that generated $Y$.
Let $I=I(\theta)$  be a user defined \textit{scalar} indicator  of information content in $Y$ about $\theta$, i.e.,\ $I$ is a non-negative
real number determined by  $f(y|\theta)$.  We can index  $Y$  by $I, $ and use the notation $Y_I$ and $y_I$ as needed. When $Y$ consists of $n$ i.i.d.  observations, we typically set $I=n$.  

Let $\mcP$ be the  set of all distributions over $\Theta$, and $D: \mcP\rightarrow [0,\infty)$
be a user defined measure quantifying the amount of uncertainty in a particular distribution or a loss function when one invokes a decision-theoretic perspective.  For example, when $\theta$ is univariate, $D$ can be the variance, the mean absolute deviation, the mean squared error to a specified
value of the parameter, etc.  The range of $D$ implies that it exists and is finite, a condition which may require us to restrict its domain to a subset of $\mcP$.
 We emphasize that our approach only requires $D$ be real valued, not that $\theta$ be univariate.
 
In general, the choice of $D$ should reflect aspects of a posterior that are most relevant to what we want to learn. A common choice is the posterior MSE  as in Section~\ref{sec:norm}:  
\begin{equation}
D_{\theta_0}(\pi(\cdot|Y)) = \Var_\pi(\theta |Y) + [\E_\pi(\theta |Y) - \theta_0]^2, \label{e:dtheo}
\end{equation}
where, for notation simplicity, we assume $\theta$ is univariate, and we use the subscript $\theta_0$ to highlight the dependence of $D$ on the true value. 
Recall that expected measures such as MSE are typically not invariant even to one-to-one transformations, regardless of whether they are for estimation uncertainty or prediction error.  \note{We stress that as a pragmatic measure to capture the impact of a prior on the actual values of these expected measures, the proposed PSS can vary with the scale of what we want to estimate or predict; see Section~\ref{s:whole}.}
We can also consider other $L_p$ measures, such as $L_1$ distance (see Section \ref{a:app} of the appendix). We discuss several possible directions in Section \ref{s:future}  on choices for $D$ as future work.
For additional ideas on choosing $D,$ see \cite{morita:thall:muller:2008} for a measure based on curvature of the log likelihood and \cite{gelman:etal:2014} for measures based on deviances. One can also use multiple $D$s to serve for different purposes.

\subsection{Define the Prior Information Function}
For a given $D$, the \textit{expected  posterior  loss (i.e., risk) with respect to the true model}  is$$
U_{\pi}(I)=\int_{S} D_{\theta_0}\left[\pi(\cdot|y_I)\right] f(y_{I}|\theta_0) \ d\mu(y_I).
$$ 
Given a baseline prior $\pi_b$, we define
$M(I)$ as the amount of information needed to match the risk  in $\pi(\theta|y_I)$ to that in $\pi_b(\theta|y_{I + M(I)})$, that is, we seek $M$ such that 
\begin{align} \label{e:Us}
U_{\pi_b}\left(I+M(I)\right)=U_{\pi}(I),
\end{align}
just as  in the normal example, where $I=k$. 
  
To see how $M(I)$, as a function of $I$, is useful for  detecting prior-likelihood discordance,   let us assume it is differentiable with respect to  $I,$ which, as an index for information, can be treated as continuous. Assuming differentiability as needed, and taking the derivative with respect to $I$ in \eqref{e:Us}, $$
U'_{\pi_b}\left(I+M(I)\right)\left[1+M'(I)\right]=U'_{\pi}(I),
$$
we arrive at,  assuming $U'_{\pi_b}\left(I+M(I)\right)\not=0$,
\begin{equation}\label{e:de}
1+M'(I)=\frac{U'_{\pi}(I)}{U'_{\pi_b}\left(I+M(I)\right)}.
\end{equation}
When $D$ is chosen appropriately, $U(I)$ should be a strictly decreasing function of $I$, since an appropriate uncertainty measure should decrease as the information $I$ increases \textit{in expectation} (see the on-line supplement for why we need to emphasize this issue, as well as \cite{meng2014got} on how variance measures violate this  monotonicity for inefficient procedures).  
This implies that the right hand side of (\ref{e:de}) will be non-negative,  yielding $M'(I)\ge -1$, confirming observation (V) from the normal example.
Moreover, a negative $M'(I)$ implies that the uncertainty decreases slower when using $\pi$ because the left hand side of (\ref{e:de}) is less than one, indicating a discordance between the likelihood function  $L(\theta|Y)\propto f(Y|\theta)$ and the prior $\pi$. The  $-1$ lower bound has a practical interpretation: the most extreme prior-likelihood conflict \textit{detectable} by $M(I)$ is when the negative information in the prior erases every single piece of information (defined by the information in a single data point) added to the likelihood.   

On the other hand, when there is no detectable discordance, e.g., when the prior $\pi$ comes from a (exchangeable) previous experiment on the same $\theta_0$, the information in the prior should stay about the same regardless of the information in the likelihood function. Hence $M'(I)$ will be approximately zero. This interpretation is most obvious when we notice that  $\lim_{I\rightarrow \infty}M(I)/I=\lim_{I\rightarrow\infty}M'(I)$ by L'H\^opital's rule if $M(I)\rightarrow\infty$, and that $I+M(I)=I[1+R(I)]$, where $R(I)=M(I)/I$. Hence $M'(I)$, for large $I$, approximates the direct measure  $R(I)$, \textit{the information gained or lost due to the prior relative to that in the likelihood.} Therefore, when the information in the likelihood grows but the prior information stays about the same, $M'(I)\approx R(I)\approx0$ for large $I$, i.e.,  the prior information is negligible asymptotically.   

In contrast, if say $M'(I)$ or $R(I)\approx - 0.5$, then  the prior-likelihood conflict has caused a reduction of  50\% information, e.g.,  our posterior mean/mode based on 1000 i.i.d. observations and our prior $\pi$ behaves like the posterior mean/mode based on 500 i.i.d. observations and  the baseline prior $\pi_b$,  which typically behaves like  the MLE based on 500 i.i.d. observations.  Clearly it is helpful  for  users of Bayesian methods  to be aware of such loss of efficiency,  just as they should be aware of the uncertainty in their  estimators.    
 
An appealing property of using $M(I)$ to measure $I_{prior}$, the prior information, is that  if we view $I+M(I)$ as the information measure for the posterior, $I_{posterior}$,  then trivially 
 \begin{equation}\label{e:id}
I_{\rm posterior}=I_{\rm likelihood}+I_{\rm prior}
\end{equation}
because the information in the likelihood, $I_{\rm likelihood}$, is $I$ in our setup. 
Whereas  $\eqref{e:id}$ is practically appealing, it is a non-standard information decomposition because $I_{\rm prior}=M $ can be negative, pointing to a prior-likelihood conflict.  

\subsection{Implementing the Asymptotic Formula}
As we will demonstrate in Section~\ref{s:th}, the normal formula (\ref{eq:normM}) holds asymptotically for a general class of distributions. Our theoretical and empirical investigations provide us sufficient confidence to suggest that, \textit{in the absence of other more reliable methods}, it can be adopted to be a rule-of-thumb for a quick assessment of the impact of the prior. Specifically, formula (\ref{eq:normM}) implies that, when $k=n$ (which is our target case),   
\begin{equation}\label{eq:key}
R(n)=\frac{M(n)}{n} = r\left[1-\left(\frac{2}{\Delta^2-1}+\frac{r}{1+r}\right)^{-1} \right].
\end{equation}
Therefore, to compute $R(n)$ we only need to compute $\Delta^2$ and $r$. Here $r$ and $\Delta$ can take on a number of asymptotically equivalent forms and hence they can be estimated in a number of different ways. 

For computational simplicity, in general, we recommend using the estimates 
\begin{equation}\label{eq:pd} \hat r=\frac{1}{d}
\ \rm{trace}\left(\hat\Sigma_{\pi_b}{\Sigma^{-1}_\pi}\right), \quad \text{and} \quad
\hat\Delta^2 = (\hat\mu_{\pi_b}-\mu_\pi)^\top{\Sigma_\pi}^{-1}(\hat\mu_{\pi_b}-\mu_\pi);
\end{equation}
where $d$ is the dimension for multivariate $\theta$. Note here we have deviated from our assumption of $d=1$ in order to provide explicit general formula, which might not be immediate for general practitioners if we only give the univariate version $\hat{r}=\hat\sigma^2_{\pi_b} /\sigma^2_{\pi}$ and $\hat\Delta^2=(\hat\mu_{\pi_b}-\mu_\pi)^2/\sigma^2_{\pi}.$
Here $\mu_\pi$ and $\Sigma_\pi$ (or $\sigma_\pi^2$) are respectively the \textit{prior} mean and variance of $\theta$ from the target prior $\pi$, and $\hat\mu_{\pi_b}$ and $\hat\Sigma_{\pi_b}$ (or $\hat\sigma_{\pi_b}^2$) the \textit{posterior} mean and variance of $\theta$ under the baseline prior $\pi_b$ (and use all the data). These four quantities are readily available for the vast majority of Bayesian analyses where a proper prior is used; note the need of assessing the prior impact relative to a baseline prior (often improper) arise typically only when the prior is proper. Just as a sanity check, for the example in Section~\ref{sec:norm}, $\sigma_\pi^2=\sigma^2_0/m$, $\hat\sigma^2_{\pi_b}=\sigma^2_0/n$, and hence $\hat r=m/n=r$ (when $k=n$). Furthermore, $\hat\mu_{\pi_b}=\bar X_n$, hence $\hat\Delta^2=m(\bar X_n-\mu_\pi)^2/\sigma_0^2$, which consistently estimates $\Delta^2=m(\mu_0-\mu_\pi)^2/\sigma_0^2$ (recall here $\sigma_0$ is a known constant). 

There are cases, however, where a proper prior does not have variance or even mean, such as the Cauchy prior in Section~\ref{s:app}. Our theory actually does not require them to exist, but rather the existence of a prior estimate of $\theta$ and the associated uncertainty measure, denoted by $\tilde\mu_\pi$ and $\tilde\Sigma_\pi$ (or $\tilde\sigma_\pi^2$)respectively. For the Cauchy prior, for example, we can use its median for $\tilde\mu_\pi$ and its scale parameter for $\tilde\sigma_\pi$. 

In those cases where the prior estimate $\mu_\pi$ and its associate uncertainty $\Sigma_\pi$ are not readily available, one can use the same routine for computing the posterior mean and variance to approximate them by applying the routine to a random selected subsample of size $n_0$. Ideally we want to set $n_0=0$,
but if that is not permissible (e.g., resulting in a nonconvergent MCMC), we can use $n_0$ the smallest possible one that still rends a well defined output from the posterior routine. That is, we are willing to move a very small part of the likelihood into the prior in order to gain computational simplification, 
and then assessing the contribution of this  enhanced prior, as an approximation to the actual prior contribution. This might cause a slightly over-estimation or under-estimation of our prior contribution, but as long as $n_0$ is a few percentages of the total $n$, the resulting estimate $\hat r_\Delta$ should serve the same purpose as that from the actual $R(n)$. Again, our practical interest is to gain a reasonably quantified feeling of the impact of the prior (e.g., whether it is $5\%$ or over $30\%$), not to pinpoint the exact prior contribution, which will not be a fruitful pursuit even if it is theoretically possible. 

\subsection{A Finite-Sample Procedure for i.i.d. Data}\label{s:diag}
Even as a quick-and-not-so-dirty assessment metric, the accuracy of (\ref{eq:key}) will depend on how soon the asymptotic kicks in.
For data consist of $X_1,\dots, X_n \overset{i.i.d}{\sim} f(x|\theta)$,  with $\pi(\theta)$ being the target prior, we can also implement $M(r)$ empirically and numerically,
as long as we are willing to perform some non-trivial computation.  Specifically,
 \begin{enumerate}
\item Choose a baseline prior $\pi_b(\theta),$ such as an objective or reference prior; see \cite{kass:wasserman:1996} and \cite{berger:bernardo:sun:2009}. A flexibility of our strategy is the allowance of atypical baselines \citep[e.g.,][]{protassov:etal:2002}.
 \item Choose $D_{\theta_0}(\cdot)$ and then construct an estimator of
 \[
 U_{\pi, \theta_0}(k) =  \E[ D_{\theta_0}(\pi(\cdot |{X}_1,\ldots, X_k)) | \theta = \theta_0], \quad k=1,\ldots, K,
 \]
where we choose $K=O(n^{1/2})$ for reasons given in Section~\ref{s:th}.  Letting  $\bomega_k$ be the $\binom{n}{k} \times k$ matrix enumerating all possible $\binom{n}{k}$ subsamples of $\{1,\dots,n\}$ of size $k$, we can then estimate $U_{\pi, \theta_0}(k)$ by  $\hat U_{\pi, \hat\theta_n}(k)$, where $\hat\theta_n$ is an efficient estimator of $\theta_0$ based on all data $\{X_1, \ldots, X_n\}$, and   
\begin{equation}\label{e:dsamp}
 \hat U_{\pi, \hat\theta_n}(k) = \frac{1}{\binom{n}{ k}} \sum_{j=1}^{\binom{n}{k}}D_{\hat\theta_n}[\pi(\cdot | X_{\bomega_k(j,1)}, \dots, X_{\bomega_k(j,k)})].
\end{equation}
  In practice,  a sub-sampling strategy, that is,  bootstrapping, will typically suffice. Obtain  $\hat U_{\pi_b, \hat\theta_n}(k)$ analogously to $\hat U_{\pi, \hat\theta_n}(k)$, with the baseline $\pi_b$ in place of $\pi$.
\item  Interpolate the $\hat U$ functions so they live on the real line.  We use linear interpolation for simplicity, but one can investigate  more sophisticated methods. We then define
$$\hat M(k) =\arg \min\{m \in \mbR:   \hat U_{\pi, \hat\theta_n}(k)  = \hat U_{\pi_b, \hat\theta_n}(m+k)\}.$$
  For $\hat{M}(k)$ to exist, we need to avoid (at least) $\hat U_{\pi, \hat\theta_n}(0)<  \hat U_{\pi_b, \hat\theta_n}(K)$, i.e., the  information in $\pi$ is so strong that it exceeds the combined  information from  the entire likelihood with  all $K$  observations  and from the baseline prior. Whereas we can try $k$ (and hence $K$) as large as $n$, the very need to do so should serve as a warning that the  prior is very informative. Indeed, if the solution still does not exist when  $k=n$,  then it suggests that at least 50\% of our posterior information will come from  our  prior  $\pi$.  

\item Plot the sequence $\hat M(k)$ and $\hat R(k)=\hat M(k)/k$ against $k$, for $k=1,\ldots, K$, and regress $\hat M(k)$ on $k$ for $k=k_0, \ldots, K$ for some suitably chosen $k_0$ to estimate an \textit{approximate limiting slope} of $\hat M(k)$ as a function of $k$, denoted by $S_K$.   Based on our current theoretical and empirical evidence, we observe the following: 
\begin{itemize}
\item  When there is no noticeable prior-likelihood discordance,  $\hat M(k)$ stays fairly constant, and hence $S_K\approx 0$, and $\hat R(k)$ will approach zero rapidly  as $k$ increases;
\item Any serious departure of $\hat M(k)$ from  being a constant function, especially as  a monotone decreasing function,   indicates  a prior-likelihood discordance;
\item Both $R(k)$ and $S_K$ serve as measures of the degree of discordance, where $R(k)$ measures the loss (or gain) due to the prior-likelihood discordance at a finite $k\le K$, and $S_K$ serves an estimator of $R(n)$, the object of our central interest,  for  $n >> K$; \item Very serious prior-likelihood conflict will cause  $\hat R(k)$ or $S_K$ to approach $-1$, i.e.,  the conflict would essentially wipe out all the information in the likelihood.     
\end{itemize}
\end{enumerate}
 
We use $S_K$ instead of $\hat R(K)$ to estimate $R(n)$ because $K$ needs to be chosen such that $n/K=O(n^{1/2})\rightarrow \infty$ and hence $\hat R(K)$ is often too far from $R(n)$.  However, as long as we are able to choose $k_0$ such that $\hat M(k)$, for $k\ge k_0$, is reasonably linear in $k$, we can approximate  $R(n)$  by the slope from  regressing $\hat M(k)$ on $k$ for $k\ge k_0$.  The theoretical and empirical evidence provided below indicates that this approximation is of practical value.
Nevertheless, we do not have any evidence, nor intuition, to  suggest that it cannot be improved; we hence invite readers to search for improvements.    
  
\section{ Theoretical Underpinning}\label{s:th}
This section establishes an asymptotic result to provide  some  theoretical insight about the  procedure given in Section \ref{s:diag}, with the $D$ being the estimated posterior MSE  given by, for $k=1,  \ldots, K (\ll n)$,
\begin{equation}
D_{\hat\theta_n}(\pi(\cdot|\vec{X}_k)) = \Var_\pi(\theta |\vec{X}_k) + [\E_\pi(\theta |\vec{X}_k) - \hat\theta_n]^2, \quad {\rm with} \quad \hat\theta_n=\E_{\pi_b}[\theta|\vec{X_n}]. \label{e:etheo}
\end{equation}
Any suitable asymptotic regime  here must  permit  the prior influence to grow in some suitable way with the likelihood data size. Otherwise the prior contribution would become negligible by design, as with the standard asymptotic framework for  the large-sample  equivalence  between Bayesian and likelihood inferences. We emphasize that the standard asymptotic framework is statistical, meaning that its limiting process is a statistically feasible one, at least conceptually. The non-standard asymptotics strategy we adopt is mathematical, invoked purely for obtaining a tractable mathematical expression to approximate a target quantity. This is in the same spirit as the popular large-$p$-small-$n$ asymptotics, where the number of parameter $p$ is assumed to grow with the sample size $n$, a process typically with no scientific or statistical reality, because nature and humans do not collaborate with each other in choosing the number of variables ($p$) relative to the sample size ($n$). 
See \cite{LiMeng:2020} for a discussion about the importance of distinguishing between \textit{mathematical asymptotics} and \textit{statistical asymptotics}.

Our non-standard regime shows that the key identity for the normal case, (\ref{eq:normM}), holds asymptotically, essentially for  all posterior-prior  families that  satisfy the following ``functional shrinkage" assumption.   For simplicity, we  restrict $\theta$ to be univariate, but the results hold generally with necessary extensions of notation, as illustrated by (\ref{eq:pd}).        

\begin{assumption} \label{a1}  Assume $X_1, \dots, X_n \overset{i.i.d.}{\sim} f(x|\theta)$ with respect to a  measure on $\mbR$, where $\theta\in\mbR$.     Assume that the prior, $\pi(\theta)$, is such that there exists $m>0$ and $\mu_{m}\in \mbR$ such that for \textit{any} $\vec X_k=\{X_1,\ldots,X_k\}$, where $k \ge k^*$ for some fixed $k^*$,  the following hold     
\begin{equation}\label{e:conj}
\E_\pi[\theta|\vec X_k] = u(T_{k,m})+O_p((m+k)^{-1})  \quad \mbox{and} \quad \Var_\pi[\theta|\vec X_k] = \frac{v(T_{k,m})}{m+k}+O_p((m+k)^{-2}),
\end{equation}
where $u$ is a twice differentiable and $v>0$ is a differentiable,  
\begin{equation}\label{e:tkm}
T_{k,m}=\frac{m \mu_m + k \bar{T_{k}}}{m+k},
\end{equation}
and $\bar T_{k}$ is the average of some $T_i = T(X_i)$ over $i=1,\ldots,k$, whose mean $\mu_T=\E[T_i|\theta]$ and variance $\sigma_T^2=\Var[T_i|\theta]$ are assumed to exist.  Furthermore, assume that our baseline prior $\pi_b$ corresponds to the limiting case of $\pi$ when $m$ is set to zero.  That is, 
\begin{equation}\label{e:conjb}
\E_{\pi_b}[\theta | \vec X_k] = u(\bar T_{k})+O_p(k^{-1})  \quad \mbox{and} \quad \Var_{\pi_b}[\theta|\vec X_k] = \frac{v(\bar T_{k})}{k}+O_p(k^{-2}).
\end{equation}
\end{assumption}

 Assumption~\ref{a1} is satisfied by many common conjugate prior distributions  including the six natural exponential families (NEFs) with quadratic variance functions \citep{morris:1982}; see the online Supplement. More broadly,  under standard regularity conditions,  a log-likelihood function resulting from i.i.d. data is known to be asymptotically  quadratic, and hence we can expect 
Assumption~\ref{a1} to hold at least asymptotically.  Perhaps the easiest way to gain insight is to consider the parallel to the normal case in Section~\ref{sec:norm}, where is particularly easy to understand the $T_{k.m}$ expression in (\ref{e:tkm}), as a weighted average of the sample mean and the prior mean, with weights proportional to their respective precisions. As is well known, this weighted average is the backbone of the much celebrated shrinkage  estimation from a Bayesian perspective \citep[e.g.,][]{efron1973stein}. Hence we view Assumption~\ref{a1} as an assumption of \textit{functional shrinkage} because it requires both the posterior mean and variance as functions of the standard linear shrinkage   estimator as in the normal example of Section~\ref{sec:norm}, where   $T(x)=x$.  

The comparison with the normal example also gives us the insight that $m$  can be interpreted in general as the \textit{nominal} PSS measured on the same unit scale as  the likelihood data size.  We say $m$ is \textit{nominal} because the real PSS  must  take into account the potential prior-likelihood discordance,  as emphasized previously. Furthermore, Assumption 1 does not require the existence of the prior mean, but only the existence of $m$ and $\mu_m$ (see   the exponential example in Section~\ref{s:expo}).  Therefore, in general, $\mu_m$ should be regarded as a measure of prior centrality, and is not necessarily the prior mean for $\mu_T$.

We use the notation $\mu_m$ to indicate that the prior centrality  can depend on  $m$. This is obvious when our prior information actually comes from a previous study based on a data set $\{\tilde X_1, \ldots, \tilde X_m\}$, which are i.i.d. samples from $f(x|\theta_{1}),$ where  $\theta_1$ may differ from $\theta_0$, the generating  value for our data $\{X_1, \ldots,  X_n\}$.   Assuming the previous Bayesian analysis used the same baseline prior $\pi_b$, we know from (\ref{e:conjb}) that the prior mean for $\mu_T$ will be approximately $\tilde  T_m$,  the average of $\{T(\tilde X_i), i=1,\ldots, m\}. $

 The simple concept that we can approximate $ \mu_m$ by $\tilde T_m$ turns out to provide rather useful  insights for forming an appropriate asymptotic regime, a regime that permits $m$ to grow with $k$ such that $k/(k+m)$ stays within the interval $(0,1)$. Specifically,  if we let  $\Delta=\sqrt{m}(\mu_m-\mu_T)/\sigma_T$,  then the fact that $\Delta \approx \sqrt{m}(\tilde T_m-\mu_T)/\sigma_T$  means that even under the assumption $\theta_1=\theta_0$,  $\Delta^2$ will not approach  zero,  because  $\Delta^2$ is a test statistic---based on data $\tilde T_m$---of the null hypothesis $H_0: \theta_1=\theta_0$;
its asymptotic null distribution, as $m\rightarrow \infty$, is the  chi-squared  distribution  $\chi^2_1$, which is exact in the normal example of Section~\ref{sec:norm}. It is therefore meaningful in our asymptotic regime to consider $\Delta$ as fixed while permitting  $m$ to grow,  because $\Delta$ provides a probabilistic yardstick for  assessing   how  the prior data set, as a proxy for the prior information, differs from the  current data set  used for the likelihood function.  
Consequently, we build our asymptotic regime under the following assumption: 
\begin{assumption}\label{a2} For the $\mu_m$ given in Assumption~\ref{a1}, we assume that it can be expressed as 
\begin{equation}\label{e:prim}
\mu_m=\mu_T+ \Delta\frac{\sigma_{T}}{\sqrt{m}}+O_{p}(m^{-1})
\end{equation}
for some fixed constant $\Delta\in \mbR.$  
\end{assumption}

As shown shortly, the two assumptions above play a critical role in establishing an asymptotic expression for  $R(k) = M(k)/k$.  The next assumption is of a technical nature to ensure that our asymptotic expression is unique, and it holds trivially  in virtually all applications. Nevertheless it is needed for eliminating  pathological cases  where properties that hold in probability, as  in  (\ref{e:conj}),  fail to hold almost surely, as required by Assumption~\ref{a3}; for practical purposes, this difference is almost immaterial.    

\begin{assumption}\label{a3}
We assume  (i) both $\hat U_{\pi,\hat\theta_n}(I)$ and $\hat U_{\pi_b,\hat\theta_n}(I)$ converge  almost surely to zero  as  $I\rightarrow\infty$, and (ii) for any finite stopping time $\hat I$,  $\hat U_{\pi_b,\hat\theta_n}(\hat I)>0$, almost surely.   
\end{assumption}

\note{}We are now ready to state our main theoretical results;  see Appendix~\ref{s:proofthm1} for proof.
   
\begin{theorem} \label{t:main} Assume $\hat D$ as defined in \eqref{e:etheo} and  that both  $k$ and $m$ increase to infinity with  $n$, with the restriction   $k=O(n^{1/2})$ and that $r=m/k$ is strictly bounded away from zero and infinity even at its limit.  Letting $c= [u'(\mu_{T})]^2\sigma^2_T/\{[u'(\mu_{T})]^2\sigma^2_T+v(\mu_T)\}\le 1$, we then have the following results.
  
\begin{description}
\item[(A)] Under  Assumptions 1 and 2, any $M(k)=kR(k)$, where 
\begin{equation}\label{e:keyr}
R(k)=R_r(\Delta^2)+O_p(k^{-1/2}),\ \text{with}\ R_r(\Delta^2)
=r\left( 1-   \frac{c(1+r)(\Delta^2 - 1)}{ (1+r) + cr(\Delta^2 - 1)} \right),
\end{equation}
gives the PSS to the order of $O_p(k^{-1/2})$, in the sense that it satisfies
\begin{equation}\label{eq:one}
\frac{\hat U_{\pi,\hat\theta_n}(k)}{\hat U_{\pi_b,\hat\theta_n}(k+M(k))}=1+O_{p}(k^{-1/2}).
\end{equation} 
 \item[(B)] Further, under Assumption~\ref{a3}, (\ref{eq:one}) holds if and only if (\ref{e:keyr}) holds. 
\end{description} 

\end{theorem}

 Expression (\ref{e:keyr}) illustrates the role of $\Delta^2$ in determining the  behavior of $R(k)=M(k)/k$. As in the normal example, when $\Delta^2=1$, $R_r(1)=m/k=r$, implying that $M(k)$ will recover the nominal PSS $m$ asymptotically.  When $\Delta^2\rightarrow \infty$, representing the extreme prior-likelihood conflict, $R_r(\Delta^2)$ goes to its lower limit $-1$; clearly $R_r(\Delta^2)$ decreases strictly monotonically to $-1$ as $\Delta^2$ increases to $\infty$.      

At the other extreme, that is, when $\Delta=0$, we see that because we can write 
\begin{equation}\label{e:ragene}
R_r(\Delta^2)=rA_r(\Delta^2), \quad \text{where}\quad A_r(\Delta^2)=1-\frac{c(1+r)(\Delta^2-1)}{1+r+rc(\Delta^2-1)},
\end{equation}
we have $R_r(0)=rA_r$, with (recall $0\le c \le 1$)
\begin{equation*}
 A_r=1+\frac{c(1+r)}{1+r(1-c)} \ge 1.
\end{equation*}
Therefore, asymptotically, $M(k)$ is larger than the nominal size $m$ by the  factor $A_r$. This is the beneficial discordance  phenomenon seen in the normal example, where $c=1/2$.  

Intriguingly,  $c=1/2$ holds for a wide range of models.  
Specifically, let us assume the usual large-sample equivalence between  the likelihood inference and the Bayesian inference  under our baseline prior $\pi_b$, that is, as $k\rightarrow \infty$, the posterior variance of $\theta$, $\Var_{\pi_b}[\theta|\bar X_k]$ is almost surely  the same as the sampling variance of the posterior mean $\E_{\pi_b}[\theta|\vec{X}_k]$.\ Then we have 
from (\ref{e:conjb}),  by the $\delta$-method,  that 
\begin{equation}\label{e:supr}
1=\lim_{k\rightarrow \infty}\frac{V\left[\E_{\pi_b}(\theta |\vec{X}_k)\big\vert \theta\right]}{V_{\pi_b}(\theta |\vec{X}_k)}= \lim_{k\rightarrow \infty}\frac{[u'(\mu_{T})]^2\sigma^2_T/k}{v(\mu_{T})/k}
=\frac{[u'(\mu_{T})]^2\sigma_{T}^2}{v(\mu_T)},
\end{equation} 
and hence the $c$ as specified in Theorem \ref{t:main} is $1/2$. In Appendix B, we  will verify that (\ref{e:supr}) holds for all models examined there. Furthermore, when $c=1/2$,   the increase in PSS due to beneficial discordance  is always an additional $(1+r)/(2+r)$ percent of information, and hence it is between 50\% and 100\%, exactly the same as in the normal example. 
 This is  an unexpected finding, especially because of its simple and general  nature.   The assumption (\ref{e:supr})  holds rather generally because  of the standard asymptotic equivalence between the likelihood and Bayesian inferences.   For some convolution families under the single observation unbiased prior \citep[SOUP; under which posterior mean is unbiased as a point estimator; see][]{meng2002single},   (\ref{e:supr})  holds exactly for any $k$.  

More generally, we see from (\ref{e:ragene}) that the beneficial  information   kicks in as soon as $\Delta^2<1$, and the amount of increased information is monotone in $1-\Delta^2$. Similarly, when $\Delta^2>1$, the amount of the information lost is a monotone increasing function of $\Delta^2-1$. 
This result also says that when $|\Delta^2-1|$ is too small, our method will not be able to   detect the prior-likelihood discordance, especially considering  we can only detect the type of  discordance with the likelihood that is not already presented  in the baseline prior,  as demonstrated   below.

\section{Empirical Illustrations} \label{s:ex}

\subsection{Computational Considerations} \label{s:comp}
To  empirically  demonstrate the performance of our procedure, we first address its computational requirements.  The brute-force implementation of the procedure outlined in Section~\ref{s:diag} can impose a substantial burden because of the need to recalculate posterior summaries (e.g., means and variances) for many different subsamples. For some models, this may simply be infeasible.  However, when 
working with conjugate families, means and variances of the posterior can be immediately calculated, and hence our methods can be carried out very efficiently.  For example, the  simulations presented below were carried out in R making heavy use of vectorization and parallelization, and each simulation study (e.g., the entire normal example)  took  5-30 minutes, depending on the sample size, on a laptop running an Intel i7 processor. 

When one is not working with conjugate families, usually posterior calculations are carried out using MCMC.  Having to obtain thousands of separate MCMC samples is often impractical.  However, this can be sidestepped by a careful use of importance weights, when they are easy to compute.  This can be very effective as the posteriors for two different subsamples are usually quite close.  Therefore, only one or a few large MCMC samples need to be generated, which can then be used for importance sampling.  To illustrate, let $(X_1,\dots, X_k)$ and $(X_1^\prime,\dots, X_{k^\prime}^\prime)$ be two separate subsamples.  If the MCMC samples, $\theta_1,\dots, \theta_{m}$, were generated using $\pi(\theta | X_1, \dots, X_k)$ but  we wish to calculate the posterior mean $\theta'_{post}=\E[\theta | X_1^\prime,\dots, X_{k^\prime}^\prime]$, 
then
we can estimate it by \[
\hat\theta_{post}'=\frac{\sum_{i=1}^{m} w_i \theta_{i}}{\sum_{i=1}^m w_i}
\quad
\text{where}
\quad
w_i \propto  \frac{\pi(\theta_i | X_1^\prime, \dots, X_{k^\prime}^\prime)}{\pi(\theta_i | X_1, \dots, X_k)}.
\]
This approach works well as long as we choose $(X_1, \dots, X_k)$ reasonably wisely  so the (sample) variance of the importance weight $w$  is not too large, because the effective Monte Carlo sample size for $\hat\theta_{post}'$ is largely \textit{bounded below}\footnote{\small We thank the AE for pointing out that this often quoted formula for $m_{\rm eff}$ can be very misleading in cases where the importance sampling weights are designed to \textit{improve} the Monte Carlo efficiency, for which cases a term neglected in deriving $m_{\rm eff}$ is in fact not negligible. } by $m_{\rm eff}=m/(1+\Var(w))$ (see \cite{kong:1992}; \cite{liu1996metropolized}), assuming the underlying MCMC chain has mixed well (and without taking into account the comparison of CPU time).   Using this approach we were able to recreate the normal example, employing parallelization and the Rcpp package (to execute the necessary loops in C++), with a computation time less than 2.5 hours (with the same equipment).  There we took an MCMC sample of size $m=100000$, 10000 subsamples, and $k$ up to 100, with  $m_{\rm eff}$ varying  from 7000 to 60000.

\subsection{A Simulation Study}\label{s:simu}
Here we provide numerical illustrations for the normal and exponential settings.  Throughout,   the estimated measure of uncertainty, $D$,   is  from \eqref{e:etheo}, and  
we take 100000 sub-samples with replacement, to construct the estimates $\hat{U}_{\pi, \hat\theta_n}$ and $\hat{U}_{\pi_b, \hat\theta_n}$.  

The goal of this section is to highlight the empirical performance of our procedure and validate the approximation formula given in \eqref{eq:key}.  To that end, we will choose the hyper parameters to reflect a desired level of $r$ and $\Delta$; in all examples below we take 100 replicates, $n=1000$, $r = 0.05$, and consider $\Delta = 0, 0.5, 1, 1.25$ to reflect differing levels of conflict between the prior and the likelihood.  To produce the ``truth" in each setting, we also use a Monte Carlo approach with 100000 draws to approximate the population values for the $U$, which can then be used to approximate the population level values for $M(k)$ and $R(k)$.  These values will be marked as dashed lines throughout the plots.

\subsection*{Example 1: Normal with Known Variance}
The true mean is $\mu_0 =1$ and the variance is $\sigma^2=1$; the latter will be treated as known.  The baseline is taken to be normal with mean zero and infinite variance (i.e.,  a  ``flat prior").   
The first row of Figure \ref{f:normkn}
depicts plots of $\hat M(k)$ for 100 replications given in grey.  Starting from the left column, the plots correspond to $\Delta= 0,0.5,1,1.25$.  A solid black line is included as the cross-sectional average of the grey lines and a dashed red line is given as the population level $M(k)$, computed via Monte Carlo. We see the solid and dash lines are practically the same.  The second row of Figure \ref{f:normkn}
provides boxplots comparing $\hat R(n)$ when using the plots to approximate $\hat M(n)$ (\textit{RPlot}) versus using the asymptotic formula (\textit{QuickEst}). When using the plots, we average the last ten values of $\hat M(k)$ and then divide by $n$ to approximate $R(n)$. 
That is
\begin{equation}\label{eq:mave}
    \hat R_{plot}(n)=\frac{\frac{1}{10}\sum_{k=K-10}^K \hat M(k)}{n},
\end{equation}
where $K=n/2 = 500$ in our simulation studies. For the normal case, the asymptotic formula is exact and thus we see that the two procedures agree as we would expect.    

Lastly, we note that the slope of the $\hat M(k)$ plots reflect the level of discordance between the likelihood and prior.  The case of ``no conflict" corresponds to $\Delta = 1$ which is the third plot and has basically a slope of 0.  There the effective sample size of the prior is about 50, which corresponds to $ r=0.05$ with $n=1000$.  As $\Delta$ moves away from 1, we see negative slopes reflecting the discordance. However, for the two $\Delta$ smaller than 1,  the discordance is beneficial because it leads to larger effective sample sizes, about 70-80, demonstrating the super information phenomena.

\begin{figure}
    \centering
    \includegraphics[width=1\textwidth]{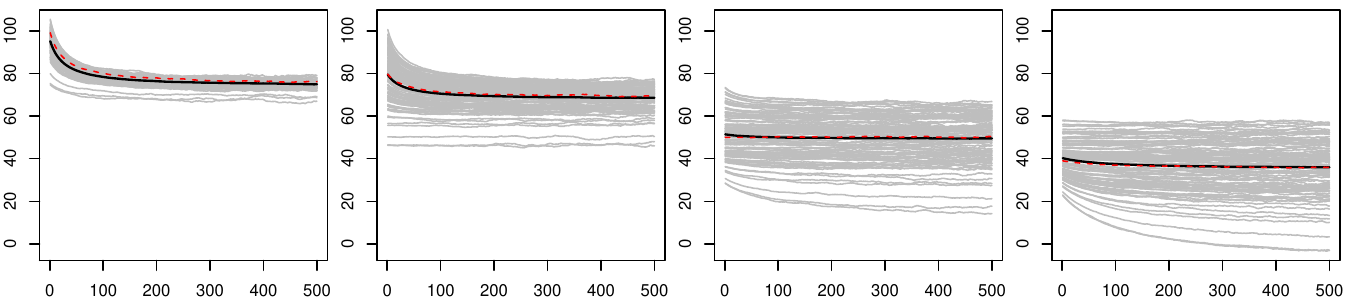}
    \includegraphics[width=1\textwidth]{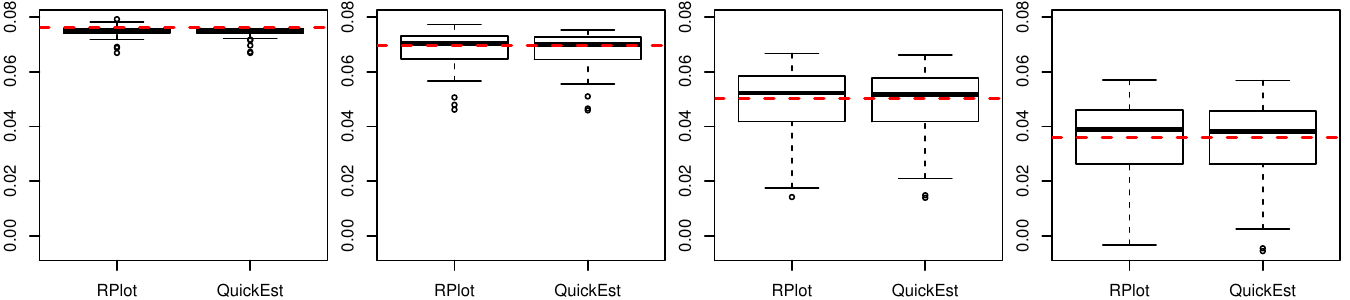}
    \includegraphics[width=1\textwidth]{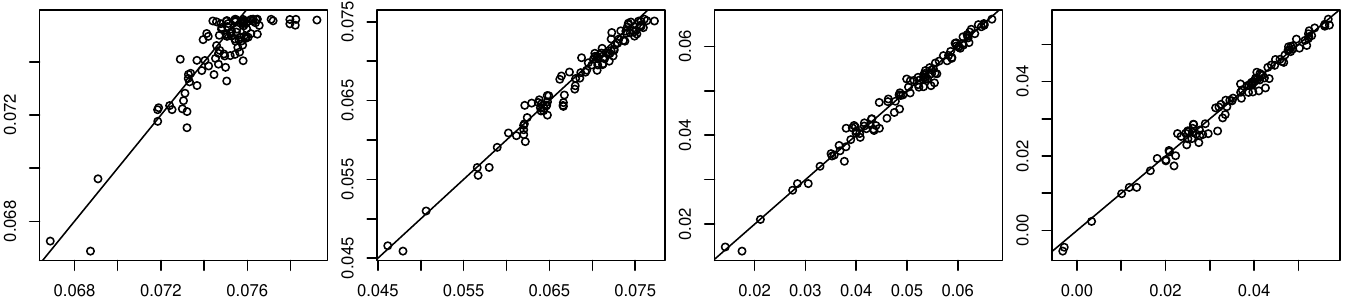}
    \caption{Normal Distribution. From the left, the four columns correspond to $\Delta=0, 0.5, 1, 1.25$ respectively. The first row displays 100 replications of the $\hat M(k)$ curve, with their point-wise average indicated as the solid line.The second and third rows compare the two estimates of $R(k)$, \eqref{eq:mave} and \eqref{eq:key}, via boxplots and scatter plots (with the QuickEst method on the y-axis).  In all plots, the red dash-line represent the estimand, which itself is approximated by using 100000 Monte Carlo samples.}  
    \label{f:normkn}
\end{figure}

\subsection*{Example 2: Exponential under  Two Parameterizations}\label{s:expo}
We now assume that $X_1,\dots, X_n$ are i.i.d. exponential random variables with mean $\mu=\lambda^{-1}=1/2$ and variance $\sigma^2 = \lambda^{-2}=1/4$.  By comparing the case of $\theta=\mu$ with  $\theta=\lambda$, we explore the nature of prior-likelihood discordance with respect to parameterizations. 
As we will see, while the two parametrizations do differ, the differences are not drastic.

\begin{figure}
    \centering
    \includegraphics[width=1\textwidth]{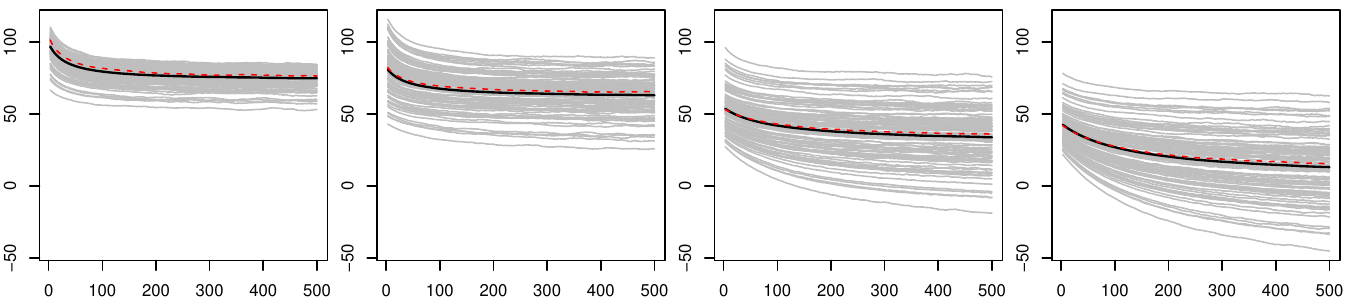}
    \includegraphics[width=1\textwidth]{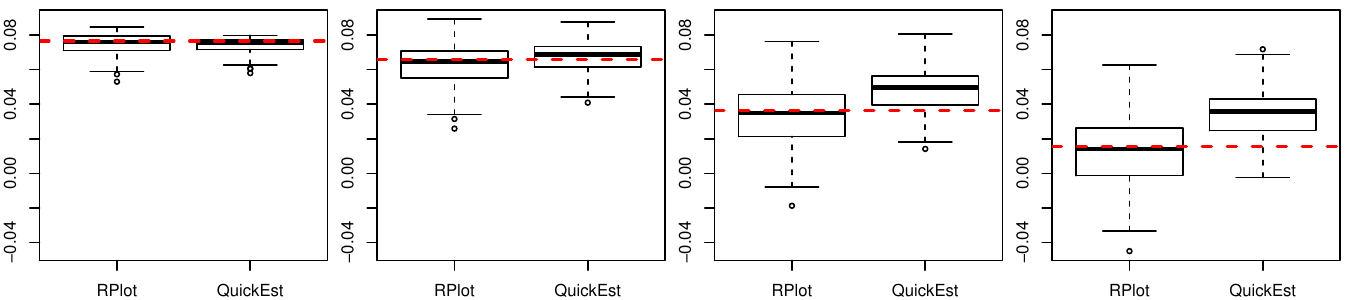}
     \includegraphics[width=1\textwidth]{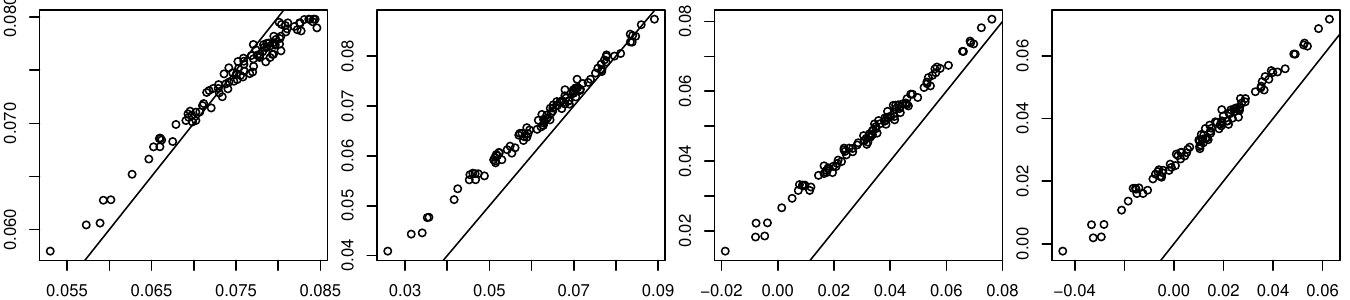}
    \caption{Exponential Distribution with $\theta=\lambda$.  The same caption for Figure~\ref{f:normkn} applies.} 
    \label{f:exp_rate}
\end{figure}

\begin{figure}
    \centering
    \includegraphics[width=1\textwidth]{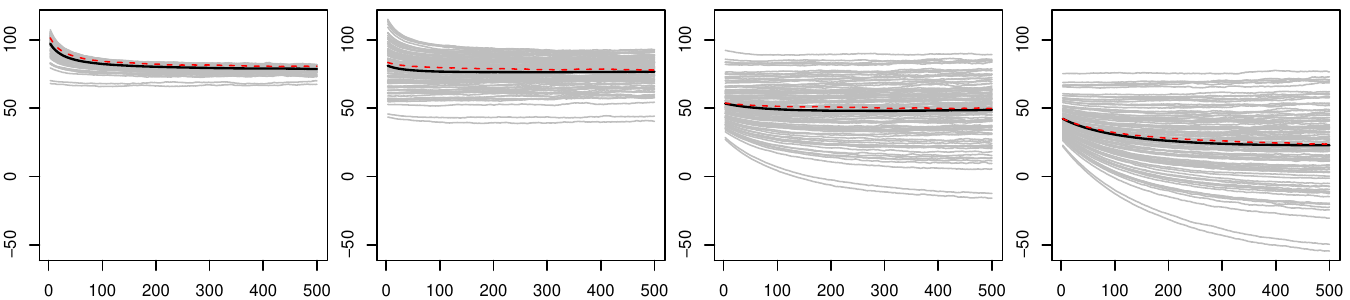}
    \includegraphics[width=1\textwidth]{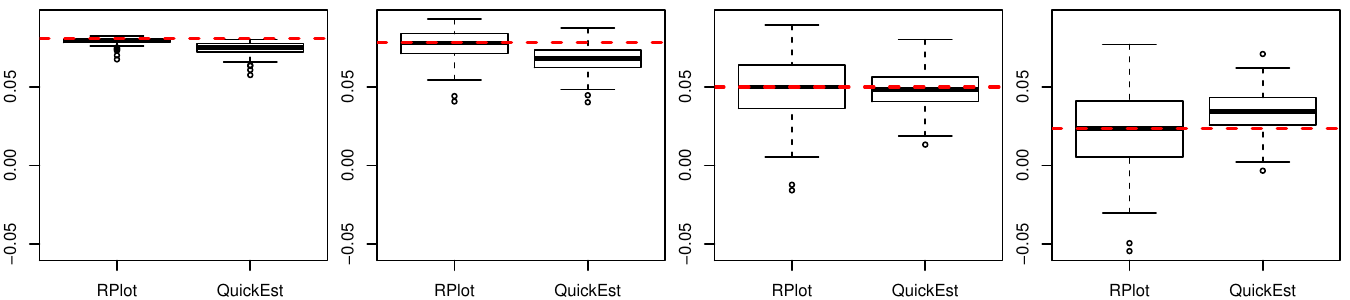}
        \includegraphics[width=1\textwidth]{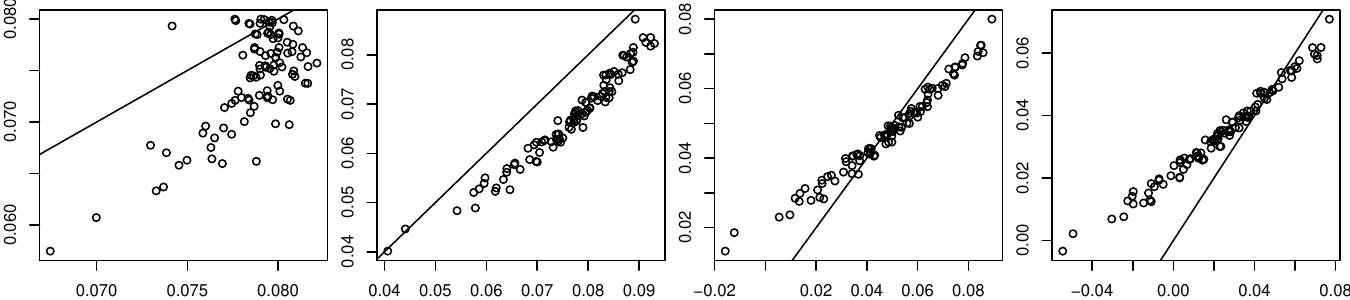}
    \caption{Exponential Distribution with $\theta=\mu$.  The same caption for Figure~\ref{f:normkn} applies. }
    \label{f:exp_mean}
\end{figure}

The  conjugate prior on $\lambda$ is gamma, $\Gamma(\alpha,\beta)$, while the conjugate prior on $\mu$ is then the inverse gamma $\Gamma^{-1}(\alpha,\beta)$.  
Our baseline is given by taking $(\alpha,\beta) \to 0$, which yields $\pi_b(\theta) \sim \theta^{-1}$, regardless of whether $\theta=\lambda$ or $\theta=\mu$, and corresponds to the Jeffreys prior for this likelihood.  The posterior of $\lambda$ under the $\Gamma(\alpha,\beta)$  prior is also a gamma distribution with mean and variance 
\begin{align}\label{e:expe}
& \E_\pi[\lambda |\vec{X}_k] = \frac{\alpha + k}{\beta+ k \bar{X_{k}}},
&\Var_\pi[\lambda |\vec{X}_k] =  \frac{\alpha + k}{(\beta+  k\bar{X}_{k})^2}. \end{align}
Similarly, the posterior for $\mu$ under the $\Gamma^{-1}(\alpha,\beta)$ prior is the  inverse gamma with mean and variance 
\begin{align}\label{e:expe2}
& \E_\pi[\mu |\vec{X}_k] = \frac{\beta+ k \bar{X_{k}} }{  \alpha + k - 1},
&\Var_\pi[\mu |\vec{X}_k] =  \frac{(\beta+ k \bar{X_{k}})^2}{(\alpha + k - 1)^2(\alpha + k - 2)}. 
\end{align}

The results for the rate and mean parametrizations are given in Figures \ref{f:exp_rate} and \ref{f:exp_mean} respectively.  The plots are analogous to the normal plots given in Figure \ref{f:normkn}.  We once again see the plots for $\hat M(k)$  are centered around the population quantity, as are the boxplots for $\hat R_{plot}(n)$.  However, unlike the normal case, the asymptotic formula is less accurate. While the distortion is still not extreme and hence the asymptotic formula can still be used as a quick approximation, it highlights the benefits of our more computationally involved algorithm.

In terms of the parametrization, there are slight differences.  In particular, the variability of our procedure is higher for the rate parameter $\lambda$ than for the mean parameter $\mu$ and our asymptotic approximation seems to work better for the mean as well.  However, the broader message is essentially the same. The asymptotic formula provides a computationally  cost-effective approximation, but the full algorithm is useful, especially as we move away from normality and linear estimators.

\subsection{Application:  Logistic Regression for Predicting Lupus} \label{s:app}
We apply our methods on a data set provided by Dr. Haas, a client at the University of Chicago's consulting program, as reported in \cite{dyk:meng:2001}.
The data set consists of 55 patients, 18 of which have membranous lupus nephritis also known as stage V lupus.  We also have measurements on the difference between immunoglobulin G3 (IgG3) and G4 (IgG4).  \cite{haas:1994} was interested in the relationship between this difference and the presence of stage V lupus.  To that end, a logistic regression model on disease status was used where a covariate representing the difference between IgG3 and IgG4 was included.  A summary of the data (in counts) is reported below.

\begin{table}
\begin{centering}
\begin{tabular}{|r|rrrrr|}
\hline
& \multicolumn{5}{c|}{IgG3 - IgG4} \\
\hline
 Lupus & 0 & 0.5 & 1 & 1.5 & 2 \\ \hline
0 &  31 &   2 &   2 &   0 &   2 \\ 
  1 &   5 &   0 &   2 &   6 &   5 \\ \hline
\end{tabular}
\end{centering}
\caption{Counts of patients with and without stage V lupus against IgG3-IgG4, the difference between immunoglobulin G3 and G4 levels. \label{t:data}}
\end{table}

\cite{gelman:jakulin:pittau:su:2008} investigated  the idea of a \textit{weakly informative prior}, and for logistic regression suggested, after standardizing appropriately, that one use a Cauchy prior with a scale of 2.5 on the slope parameter.  We use our methodology to explore how weak or strong such a prior really is.  We evaluate candidate priors from the Cauchy distribution with scales 2.5, 5, and 10 against a baseline prior,  taken to be Cauchy with a scale of 10000.  The results end up being fairly robust against the choice of baseline, as we also tried Cauchy  with scales 100 and 1000, as well as a normal baseline, as reported in Section \ref{a:app} of the online supplement.  We use the metric \eqref{e:dsamp} to compare the two priors.  However, instead of taking the mean of this metric over subsamples, we take the median to combat the well-known small-sample instability of logistic regressions, and for the same reason, we examine $\hat M(k)$ only for $ k > 10$.  (To be more consistent with the median approach, in the online supplemental we also explore using the mean absolute deviation for $D$ instead of the MSE, however the results are nearly identical.) To reduce the impact of the nonlinear part of $\hat M(k)$ on the estimation of $R(55)$, we take $k_0=20 $ to approximate $R(55)$ by the least squared estimator based on $k=20, \ldots, 35$.  

\begin{figure}[ht]
\begin{center}
\includegraphics[width=0.33\textwidth]{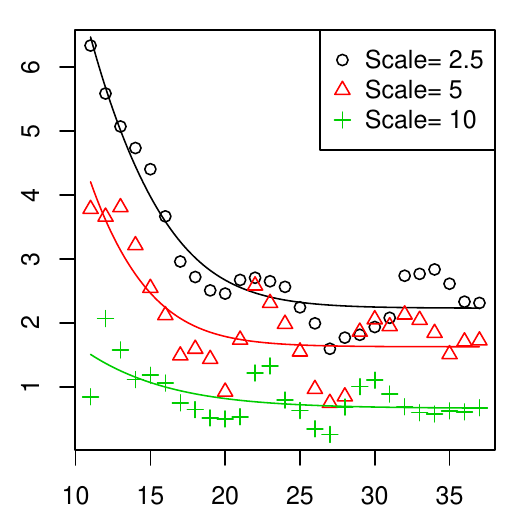}
\vspace{-.2in}
\end{center}
\caption{Plot of the estimated PSS $\hat M(k)$ for the  application in Section \ref{s:app}.  Estimated slopes are $-0.1173$,  $-0.0646$, and $-0.0283$ using $k>10$, for scales $2.5, 5$, and $10$ respectively. The estimated slopes become 0.0032,  -0.0024, and -0.0079 respectively when using  $k>20$. \  \label{f:app}}
\end{figure}

The results are plotted in Figure \ref{f:app} with the slopes given in the caption, using both $k_0=10$ and $k_0=20$.  The plots are a bit more chaotic than in our simulations due, likely, to the  aforementioned instability of logistic regression with small sample sizes, but none of them suggests more than  6  PSS. The prior suggested by \cite{gelman:jakulin:pittau:su:2008}, that is, with scale=2.5, seems to indeed depict a \textit{weakly informative prior},   equivalent to between 2  and 6 data points, which is  no more than $10\%$ of the information provided by the likelihood.  There might be some small amount of prior-likelihood discordance.  By taking the scale up to 5 or 10, the discordance is reduced, and so is the prior impact.  Indeed, the slope estimators based on $k> k_0={20}$
are essentially zero regardless of the scale, indicating essentially negligible  prior impact  with $n=55$.  Such practical, quantifiable, and  interpretable assessments can help greatly to strengthen our inferential conclusions and to communicate them convincingly, by reducing both the impact and the appearance of ad hoc choices made during our inference process. Evidently it is more scientific to numerically demonstrate that the impact of a prior is no more than adding $10\%$ of data than to simply declare that it is weakly informative.
For more studies on weakly informative prior, see \cite{gelman:2006} and \cite{polson:scott:2012}. 

We remark that logistic regression is a telling example about why it is important to formulate  PSS  as a measure relative to the likelihood function, and hence it is data dependent.
It is well-known that when the observed data exhibit a (nearly) perfect separation pattern, i.e.,  when a predictor  (nearly) perfectly separates those with positive outcome from those with negative outcome,  we will run into a (nearly) non-identifiability issue. This issue occurs rather frequently in practice, and Bayesian methods have been suggested as an effective way to address the problem, as detailed in \cite{gelman:jakulin:pittau:su:2008}, and more recently in \cite{rainey2016dealing}, who  concluded  that  ``When facing separation, researchers must \textit{carefully} choose a prior distribution ...." (emphasize is original). But since separation is a data-dependent phenomenon, and the prior is brought in to combat the issue that the likelihood function is too flat (e.g., MLE does not exist) to form a proper posterior without a suitably informative  prior,  the relative contribution of the prior information must be data dependent. Furthermore, to be meaningfully ``careful" in choosing such a prior, we need to be able to quantify the amount of contribution of various choices and hence one can judicate based on quantifiable evidence, which should also help to communicate---and hence to ensure the reproducibility of---our findings.

\section{Responding to Reviewers: Old Contemplation and New Explorations} \label{s:disc}

\subsection{Why do we still work on such an old problem?}\label{s:old}

One reviewer expressed a strong disbelief that the Bayesian paradigm had been promoted for so long without having had settled on such a basic assessment about prior influence. We were skeptical as well when we started this project, but the progress we made (since 2012) has shown us why quantifying 
prior information is fundamentally as problematic as the use of an improper prior, a matter of ongoing debate. Indeed, another reviewer questioned our use of an improper prior as the base prior. They asked if one does not have a proper prior to start with, then should one question the use Bayesian methods in the first place? Philosophically, we agree with the reviewer,  because it is a mathematical fact that the Bayesian paradigm cannot handle complete ignorance \citep[e.g.,][]{martin2016inferential}. Practically, we also agree that whenever meaningful, we should use a proper prior as the base, which is not restricted by our proposed framework in any way.

There are inferential paradigms that can quantify complete ignorance in fully logical and mathematical ways, such as belief functions \citep{dempster1968,shafer1976}. However, a recent investigation \citep{gong2020judicious} reveals that there is currently no known inferential paradigm that can do so without having to pay a ``leap-of-faith" price somewhere else (e.g., trading indeterminacy in updating rules for that in prior specifications). The implication is that we have to live with imperfection or even paradoxes one way or anther. Using improper priors seems to be the least problematic for now, given its wide-spread practice, which does not justify its use, but it does suggest that we know much more about its pros and cons, or at least know where to look for source of troubles. The following is such an example.    

\subsection{A Prior Sample Size Paradox?}\label{s:para}

We are grateful to yet another reviewer for raising an intriguing question regarding what should be viewed as nominal prior size in Example 2 discussed above. The reviewer pointed out that, inspecting the density function based on an i.i.d. sample with size $n$, we would see a term $\lambda^{n}$ in  the exponential likelihood. When this is  compared to the term $\lambda^{\alpha -1 }$ in the conjugate Gamma prior,  it seems natural  to define the  ``nominal  prior size" as $n_0=\alpha -1$, instead of what we defined as $m=\alpha$. Furthermore, since $n_0=m-1$, with our definition,
we seem to run into a negative nominal prior size when $m=0$.  This is undesirable because as a  \textit{nominal} prior size, the worst it can be should be zero, representing   no prior information. A negative nominal prior size would suggest  that we know \textit{a priori} the degree of prior-likelihood conflict before seeing any data, which would be an odd position to take, if not illogical.     

However, if we follow the reviewer's definition,   then our prior mean would be  $\E(\lambda)=(n_0+1)/\beta$. This would imply that even when the prior size is zero, meaning  there is no prior information
whatsoever, we would still have a finite ``prior mean" $\E(\lambda)=1/\beta$. This  seems  at least as illogical as having a negative prior sample size. This dilemma is rooted exactly in the indeterminacy  of an ``ignorant prior" discussed in Section~\ref{s:old}: if the prior is the posterior obtained from a previous study of sample size $n_0$, what was the prior that led to that posterior? If that was the  constant prior on $\lambda$, then that posterior would be $Gamma (n_0+1, \beta)$, and hence $\alpha=n_0+1$, consistent with the reviewer's suggestion. However, if that prior was Jeffreys prior $\pi(\lambda)\propto \lambda^{-1}$, which the reviewer also suggested to be a natural choice, then that posterior would be 
$Gamma(n_0, \beta)$, and hence $\alpha=n_0$, as we formulated.  Hence  ``looking directly at the likelihood form" and ``using Jefferys prior" cannot be   ``natural" simultaneously.

 Although the meaning of ``natural"  is debatable  in this context, we see a more compelling reason to adopt the latter,  since there are various theoretical justifications  for using the Jeffreys prior \citep{gutierrez1997exponential}. The former runs into the danger of mixing the form of $\lambda$ arising from normalizing constant for the \textit{sampling distribution} with the form of $\lambda$ arising from modeling it directly. As far as we are aware of, there is no theoretical justification on why the meanings of the powers in these two different usages of $\lambda$ should be  the same.    Of course, the central difficulty here lies in the fundamental impossibility of using a probabilistic distribution to represent ignorance, and hence it further illustrates the necessity of choosing a baseline when we assess  prior contribution.  

\subsection{Should we always check likelihood and prior, regardless of prior contribution?}\label{s:whole}

The answer to this reviewer's question of course is a resounding YES. We should always worry about the inadequacy of any part of our model, even or especially when we cannot check it. As many would argue, rightly, checking the likelihood is more important than assessing the prior and it should be done first, because it is the likelihood that permits us to make (Bayesian) inference from data to parameters. The prior serves typically an important but nevertheless supplemental role, except when we have very little information from our likelihood. But it is exactly because of the perceived supplemental role of our prior that we need to have some reasonable ways to assess the actual impact of a posited prior relative to the likelihood contribution, even if just for the purposes of calibrating with our original expectations.  For example, in the logistic model in Section~\ref{s:app}, we surely should check the adequacy of the logistic model first.  But once it is adopted for whatever reasons (e.g., convenience), then if our intention is to use weak priors for its parameters, we should at least to check whether these ``weak" priors actually have weak impact; see the mortality example in \cite{gelman1996posterior}, where a seemingly innocent uniform prior on convex curves turned out to be very influential.

This reviewer also emphasized the need to check data-prior conflict as a \textit{falsification of the prior}, an important modeling step regardless of the need to assess the prior contribution. We also agree, without getting into the debate about checking ``subjective priors" versus ``objective priors". Indeed, the whole industry of prior predictive and posterior checks \citep[e.g.,][]{box:1980,rubin1984,meng1994posterior, gelman1996posterior} were designed for such purposes, though as the reviewer noted there are multiple complications.  

First, essentially all checks are local in the sense that they can detect only some model defects in the likelihood, prior, or both. This is due to the necessarily limited capacity of the checking/testing statistics or more generally ``realized discrepancy" \citep{gelman1996posterior}. We view this locality a feature rather than a deficiency, because an almighty test would or at least should reject essentially all models, because ``all models are wrong."  George Box's mantra ``but some are useful" reminds us that our job---through judicious choices of assessments---is to ensure the relevant parts are usable. 

Second, not all model defects are consequential for the substantive questions at hand. Again, we avoid this issue by choosing the measure of uncertainty directly reflective of the analysis of interest, as emphasized in Section~\ref{s:meth}. We also agree with the reviewer that even when a defect is inconsequential for one study, it may still be useful to understand it since it can be very consequential for another study using the same model.

\note{
Such choices, however, lead to a third issue. Our uncertainty measure is not invariant to reparametrization, which is the case for the quadratic measure. Whereas we agree with the reviewer that an invariant measure has some general appeal, our proposed methods would not have much practical impact if we do not consider
measures such as MSE, which are well understood and most commonly adopted for good reasons.}

\subsection{Limitations and Future Work}\label{s:future}

Although our method has a number of appealing properties, much more needs to be done. Perhaps  the most important extension  is for problems where sample size is not a good indicator of information, as is typically the case with time series and spatially dependent data.  We   also need to establish theoretical results for scenarios that go beyond those covered in Section~\ref{s:th}, and more critically to cases where the likelihood itself is misspecified in consequential ways. \note{A reviewer also reminded us to study the issue of assessing likelihood-prior combination that could lead to substantial bias, in the sense of creating regions of parameter space that are highly probable \textit{a priori}; see \cite{baskurt2013hypothesis,evans2019measuring}.}  Applications to high-dimensional and/or non-parametric problems are another important direction to explore, and the growing literature on the relationship between prior and posterior concentrations (see for example \cite{pkv:2014} and \cite{strawn:2014} and references therein) may provide some theoretical insight on this exploration.   

In applying our method, we also encountered three practical  problems.  The first is the computational demand.   
Seeking effective computational  strategies is an area of much needed research, and the importance sampling approach presented in Section~\ref{s:comp} merely is a starting point. The second issue involves instability with small $k$. We did not encounter any problem for our simulation studies, where conjugate priors were used.  However, for the lupus nephritis application, we had to avoid small $k$  because logistic regressions can be very unstable for small sample sizes.  Any model which has stability problems for small samples can generate similar  issues.  We found switching the means to medians in our resampling scheme helped, but obviously this creates a discrepancy between the application and the current  theoretical results, which are mean-based, that is, using the $L_2$ norm. Extending our theoretical results to cover other norms, especially the $L_1$ norm, as well as more general choices of the discrepancy or uncertainty  measure $D$ is another direction for future research. Third, we need to search for more reliable estimate of $R(n)$, the relative gain or loss corresponding to the actual data size, as our current extrapolation via the slope of the PSS curve $M(k)$ is more of exploratory  nature. 

A reviewer reminded us that a particularly interesting direction for choosing $D$ involves moving to a prediction based uncertainty measure.  This can help, to a degree, with the parametrization problem as it fixes the scale of the outcome as default.  However, there are at least a few options as to how to construct a prediction based measure.  One possibility is to take a similar approach to the MSE measure we introduced, which involves conditioning on the true underlying parameter.  Another option could be based on the posterior predictive distribution and not conditioning on the true parameters, while yet another option would be similar to cross-validation, where observations not included in the $\vec Y_k$ could be used for evaluating prediction. However, preliminary explorations using cross-validation idea, as in the on-line supplement, are not very encouraging.  In particular, any newly proposed measure for $D$ has to reasonably quantify the bias of the estimates.  As demonstrated in Section \ref{a:s:D} of the Appendix, if this is not done, then results are quite unreliable.  

Finally, we can explore other methodological applications using the idea of assessing discordance via monitoring  $M(I)$.  For example, we can compare two subjective priors constructed by two different investigators, and determine whether  one has more serious discordance with a likelihood function than the other. Or perhaps we can convert  this diagnostic tool into something helpful in selecting a prior via tuning  the measure we proposed, as a \textit{functional} of a candidate prior, according to some sensible criterion. For instance,   we may want our prior to be weakly informative in the sense that the PSS should not exceed, say, 10\% of the likelihood sample size, for a chosen purpose. 

 Going even further, we can extend the idea of comparing two priors to comparing two likelihood functions, by using a common baseline prior.  If one of the likelihood models is saturated, then the conflict between them can be viewed as a misspecification of the other, unless we just have very bad luck. Of course, whether we assess prior-likelihood discordance or misspecification of a likelihood, our general goal  is the same: to be an informed Bayesian, or more generally, an informed statistical analyst.       
 
\section*{Acknowledgements} We thank Andrew Gelman for  insight comments and help regarding Section~\ref{s:app}, Murali Haran for his help in Section \ref{s:comp},  David Jones and Steven Finch for very helpful proofreading,  and a good number of reviewers for their  comments that have led to a  much improved paper.  We also  thank  the U.S. National Science Foundation, National Institutes of Health, as well as the John Templeton Foundation  for partial financial support.
{\baselineskip=14pt
\bibliographystyle{rss}
\bibliography{MeNiRe,inference_final}}

\nocite{evans:1997}
\nocite{evans:moshonov:2006}
\nocite{evans:jang:2011}
\nocite{berger:bernardo:sun:2009}
\nocite{kass:wasserman:1996}

\nocite{lee:1990}

\appendix
\clearpage
\setcounter{page}{1}
\begin{center}
\bf{\large Online Supplemental Material}
\end{center}

\section{Proof of Theorem \ref{t:main}}\label{s:proofthm1}

To prove Theorem \ref{t:main}, we will need the following Lemma on the asymptotic representation of the $\hat U$'s; proof of lemma is given in  Appendix \ref{s:proofs:lem}.

\begin{lemma}
\label{l1}
Suppose $k=O(n^{1/2})$ and $r=m/k$ is strictly between zero and infinity even at its limit as $n\rightarrow\infty$. 
Then, under the Assumptions~\ref{a1} and \ref{a2}, we have: 
\begin{equation}\label{e:expa1}
\hat U_{\pi, \hat\theta_n}(k)= \frac{\alpha}{k}+O_p(k^{-3/2}),\ \text{with}\ \alpha=\frac{1}{1+r}\left\{v(\mu_{T})+\sigma_{T}^2[u'(\mu_{T})]^2 \frac{1+r\Delta^{2}}{1+r}\right\},
\end{equation}
and \begin{equation}\label{eq:bterm}
\hat U_{\pi_b,\hat\theta_n}(k)=\frac{\beta}{k}+O_p(k^{-3/2}),\quad \text{where} \quad \beta = v(\mu_{T})+\sigma_T^2[u'(\mu_T)]^2.\\
\end{equation}
\end{lemma}

\begin{proof}[of Theorem~\ref{t:main}]
 
 By Lemma~\ref{l1}, expression (\ref{eq:one}) is equivalent to 
\begin{equation}\label{e:aslo}
[1+R(k)][\alpha +O_p{(k^{-1/2})})]=[1+O_p(k^{-1/2})][\beta+[1+R(k)]^{-1/2} O_p(k^{-1/2})]. \end{equation}
We can  verify that (\ref{e:aslo}) is equivalent  to   (\ref{e:keyr}) as long as $\liminf_{k\rightarrow\infty}[R(k)+1]>0$ almost surely, a condition which is necessary    because  otherwise we cannot write $[1+R(k)]^{-1/2} O_p(k^{-1/2})=O_p(k^{-1/2})$, which is needed for the equivalence. 

When $R(k)$ is given by (\ref{e:keyr}),  $\liminf_{k\rightarrow\infty}[R(k)+1]>0$ holds almost surely. This follows because,
otherwise, with positive probability, say $p>0$, there exists a subsequence  $\{k_i, i\ge 1\}$ such that $k_i\rightarrow \infty$ and $1+R(k_i)\rightarrow 0$. But $1+R(k_i)=1+R_{r_{i}}(\Delta^2)+\epsilon_i$, where $r_i=m/k_i$ and $\sqrt{k_i}\epsilon_i=O_p(1)$. Consequently we know with probability $p>0$, $\epsilon_i$ converges to $-(1+R_{r_\infty}(\Delta^2))<0$, where $r_\infty=\lim_{i\rightarrow\infty} r_i$.  Therefore, with probability $p$,  $|\epsilon_i|$ will be bounded away from zero when $i$ is large enough, hence it is impossible for  $\sqrt{k_i}|\epsilon_i|$ to be bounded away from infinity as $k_i$ goes to infinity. This contradicts the fact that  $\sqrt{k_i}\epsilon_i=O_p(1)$. This proves assertion (A).

To prove (B), we need  Assumption~\ref{a3}.   Again we prove this by assuming  with probability  $p>0$, the subsequence  $\{k_i, i\ge 1\}$ defined above exists.
Then for such subsequences the left hand side of (\ref{e:aslo}) goes to zero. But the right hand side can have the zero limit only if $\sqrt{k_i+M(k_i)}=-\epsilon_i/\beta$, where $\epsilon_{i}=O_p(1). $
This means with positive probability (possibly smaller than $p$), $\hat I=\limsup_{i\rightarrow\infty}[k_i+M(k_i)]$ is finite. Hence
with  positive probability  $\liminf_{k\rightarrow\infty}\hat U_{\pi_b,\hat\theta_n}(k+M(k))>0$ under Assumption~\ref{a3}(ii) because $\Pr(U_{\pi_b, \hat\theta_n}(\hat I)>0)\ge \Pr(\hat I<\infty)>0. $  But this contradicts   (\ref{eq:one}) because its left hand side then will go to zero with positive probability for the same reason as above,
yet its right hand side will go to 1 with probability one.  \quad \ Q.E.D. \end{proof}

\section{Proof of Lemma~1}\label{s:proofs:lem}

 \begin{proof}
 For notational simplicity, we abbreviate $\hat U_{\pi, \hat\theta_n}(k)$ and $\hat U_{\pi_b, \hat\theta_n}(k)$ as $\hat U(k)$ and $\hat U_b(k)$ respectively.
Under the assumption that $r=m/k$ is bounded away from zero and infinity, $k$, $m$ and $l=k+m$ are of the same order, hence we can use them  exchangeably when using the $O$ notation.  Let $\delta_k=\sqrt{k}(\bar T_k-\mu_{T})$ and $d_m=\sqrt{m}(\mu_m -\mu_{T})$,  then $\delta_k$ is $O_p(1)$  by the central limit theorem and $d_m=\sigma_T\Delta+O_p{(m^{-1/2})}$ by Assumption~\ref{a2},   and hence  
\begin{equation}\label{e:dkm}
\delta_{k,m}\equiv T_{k,m}-\mu_T=[\sqrt{k}\delta_k+\sqrt{m}d_m]/l=O_p(k^{-1/2}).
\end{equation}
Consequently $v(T_{k,m})-v(\mu_T)=O_p(k^{-1/2})$ by a one-term Taylor expansion.  Assumption~\ref{a1} then allows us to write  
\begin{equation}\label{eq:tv}
\Var_\pi[\theta |\vec{X}_k]=\frac{v(\mu_T)}{l}
+O_{p}(k^{-3/2}).
\end{equation}
For the bias term $B=\E_\pi(\theta |\vec{X}_k)- \E_{\pi_b}(\theta |\vec{X}_n)$, we expand $u({T}_{k,m})$ in (\ref{e:conj}) around $\mu_{T}$ to obtain  
\begin{equation}
\E_\pi(\theta |\vec{X}_k)= 
u(\mu_T)+u'(\mu_{T})\delta_{k,m}+O_{p}(k^{-1});
\quad \E_{\pi_b}(\theta |\vec{X}_n)=u(\mu_{T})+O_{p}(n^{-1/2}). 
\end{equation}
Only one term expansion of $\E_{\pi_b}(\theta |\vec{X}_n)$ is needed because $O_P(n^{-1/2})=O_p(k^{-1})$ under our assumption.  Consequently, we have  
\begin{equation}\label{eq:bsqu}
B^2=\left[u'(\mu_{T})\delta_{k,m}+O_{p}(k^{-1})\right]^2
=[u'(\mu_{T})]^2\delta_{k,m}^2+O_{p}(k^{-3/2}).
\end{equation}
But 
\begin{equation}\label{e:delta}
\delta_{k,m}^2=l^{-2}[\sqrt{k}\delta_k+\sqrt{m}d_m]^2=l^{-2}[k\delta_k^2+2\sqrt{km}\delta_kd_m
+md_m^2].
\end{equation}
From (\ref{eq:tv}) and (\ref{eq:bsqu}), we see that when we take a bootstrap sample of $\hat D$ of (\ref{e:dtheo}) to obtain (\ref{e:dsamp}),  to an error order of $O_p(k^{-3/2})$, it amounts to replacing the $\delta_k^i\equiv \delta^i_k (\vec X_{k})$  term  in (\ref{e:delta})  by its bootstrap average $\hat\delta_k^i \ (i=1,2)$,
which is defined similarly as in (\ref{e:dsamp}). Because $\hat\delta_k=\sqrt{k}(\bar T_n-\mu_T)$, it  differs from its mean, that is, zero, by an order of  $\sqrt{k}O_p(n^{-1/2})=O_{p}(k^{-1/2})$. Hence the middle term on the rightmost hand side of (\ref{e:delta}) can be dropped without introducing more than an error of order $l^{-2}\sqrt{m}O_{p}(1)=O_p(k^{-3/2})$, which is of the same order as the error term in (\ref{eq:tv}) or in (\ref{eq:bsqu}). 

For the $\hat\delta^2_k$ term, we will need to use some standard results for U-statistics (e.g., see Ch. 3 of \citet[][]{lee:1990}). Let $h(X_1, \ldots, X_k)=(X_1+\ldots+X_k)^2/k$. Then $\hat\delta^2_k$ is exactly the U-statistics generated by the kernel $h$, with $X_i=T_i-\mu_T$.  Therefore  it is known  that 
\begin{equation}\label{eq:uvar}
\Var(\hat\delta_k^2)\le \frac{k}{n}\Var[h(X_1,\ldots, X_k)]
=\frac{k}{n}\sigma^4_T (2+\frac{\kappa_T}{k}),
\end{equation}
where $\kappa_T$ is the kurtosis of $T_i.$   This implies that  asymptotically the  $\hat\delta^2_k-\E[\hat\delta^2_k]$ is controlled by the order  $O_p((k/n)^{1/2})=O_{p}(k^{-1/2})$. Therefore, as before, replacing $\hat\delta^2_k$ by $\E[\hat\delta^2_k]=\sigma^{2}_{T}$ in (\ref{e:delta}) introduces an error of order controlled by $l^{-2}kO_{p}(k^{-1/2})=O_p(k^{-3/2})$, no more than what is already permitted by (\ref{eq:tv}) or (\ref{eq:bsqu}).  Expansion (\ref{e:expa1}) then follows because from Assumption~\ref{a2}, $l^{-1}d_m^2=
 l^{-1}[\sigma^2_T\Delta^2+O(m^{-1/2})]=
 l^{-1}\sigma^2_T\Delta^2+O(k^{-3/2}).$

The derivation above clearly is valid when we start it by setting $m=0$, and hence $r=0$ and $\delta_{k,m}\equiv \delta_k$ (then $\Delta$ is immaterial), but this is exactly the proof needed for  (\ref{eq:bterm}).
\end{proof}

\section{Verifying Theoretical Assumptions and Results}\label{s:examples}
This section presents several prior-likelihood examples that
satisfy Assumptions \ref{a1}-\ref{a3}, and verifies the conclusions given in Section~\ref{s:th}.  All of our examples form  conjugate prior-likelihood  pairs with the following exponential forms: $X_i$ has a density of the form
\begin{equation}\label{e:dens}
f(x| \theta) = \exp\{T(x) \eta(\theta)  + \xi(\theta) +B(x) \},
\end{equation}
and the prior is a two-parameter conjugate family
(which includes the NEFs of \cite{morris:1982})
\begin{equation}\label{eq:conj}
g(\theta; a, d) = \exp \left\{a d \eta(\theta) + d \xi(\theta) +\zeta(\theta) + C(a,d)  \right\}.
\end{equation}
As before, letting $\vec{X}_n=\{X_1,\ldots, X_n\}$ denote an independent and identically distributed  sample from (\ref{e:dens}), we then have  that the posterior is proportional to
\begin{align*}
p(\theta |\vec{X}_n) & \sim  \exp\{(ad + n \bar T) \eta(\theta) +(d +n) \xi(\theta) +\zeta(\theta) \}\\
 & = \exp \left\{ \left(\frac{ad + n \bar T}{d+n}\right) (d+n) \eta(\theta) + (d +n) \xi(\theta)  + \zeta(\theta) \right\},
\end{align*}
where $\bar T = (1/n) \sum T(X_i)$.  Therefore
\[
p(\theta |\vec{X}_n )  = g\left(\theta; \frac{ad + n \bar T}{d+n}, d+n \right). \]
This means  (\ref{e:conj}) and (\ref{e:conjb}) hold with $m=d$ and $\mu_m=a$ if for the $g(\theta; a, d)$ family we have
\begin{equation}\label{eq:very}
\E(\theta)=u(a)+O(d^{-1}) \quad\text{and}\quad
\Var(\theta)=\frac{v(a)}{d}+O(d^{-2}),
\end{equation}
where $u(a)$ and $v(a)$ satisfy the properties given in Assumption~\ref{a1}. Below we show this is the case for four common applications, where the expressions of posterior  means and variances will also make it transparent that  Assumptions~\ref{a3}(i) is a consequence of the strong law of large numbers.  We therefore need only to verify Assumption~\ref{a3}(ii). \  Note Assumption~\ref{a2} is a restriction on the hyper-parameters in our asymptotic regime, and hence it is satisfied whenever we treat the value $\Delta=m(\mu_m-\mu_T)^2/\sigma^2_T$ as fixed when we let $m$ vary.

\subsection*{Exponential}
Assume that $X_1,\dots, X_n$ are exponential random variables, and hence $\mu_X=\lambda^{-1}$ and $\sigma^2_X=\lambda^{-2}$.  The conjugate prior on $\lambda$ is the gamma distribution with parameters $\alpha$ and $\beta$, $\Gamma(\alpha,\beta)$.
The baseline is given by taking $(\alpha,\beta) \to 0$, yielding $\pi_b(\lambda) \sim \lambda^{-1}$, the Jeffreys prior.  The corresponding posteriors  are respectively  gamma distributions with
\begin{align}
& \E_\pi[\lambda |\vec X_n] = \frac{\alpha + n}{\beta+ n \bar{X}_n},
&& \Var_\pi[\lambda |\vec{X}_n] =  \frac{\alpha + n}{(\beta+ n \bar{X}_n)^2}; \label{e:meane} \\
& \E_{\pi_b}[\lambda |\vec{X}_n] = \frac{1}{ \bar{X}_n},
&& \Var_{\pi_b}[\lambda |\vec{X}_n] =  \frac{1}{ n \bar{X}_n^2}. \label{e:vare}
\end{align}
It is easy to see from the first expression of (\ref{e:meane}) that for $\theta=\lambda$, we should take  $a = \beta\alpha^{-1},$  $ d = \alpha$, and $T=X$.\ Condition  (\ref{eq:very}) then is satisfied by $u(a) = a^{-1}$ and $v(a) = a^{-2}$ exactly without the $O$ terms because $\Gamma(\alpha, \beta)$ has mean and variance $\alpha\beta^{-1}$ and $\alpha\beta^{-2}$, respectively.   Assumption~\ref{a3}(ii) follows trivially from (\ref{e:vare}) because it shows that $\Var_{\pi_b}[\lambda |\vec{X}_{\hat I}]>0$ for any finite $\hat I. $
Condition  (\ref{e:supr}) can also be verified   directly from  $v(\mu_X) =\mu_X^{-2}=\lambda^2$, and  $[u'(\mu_X)]^2\sigma_X^2 = [-\mu_X^{-2}]^2\sigma^2_X =\lambda^{2}.$

Using the alternative parameterization $\mu_X$, conjugate family then becomes the inverse gamma, also with parameters $\alpha$ and $\beta$.  The posterior mean and variance functions then become
\begin{align}
& \E_\pi[\mu_X |\vec X_n] = \frac{ \beta+ n \bar{X}_n  }{\alpha + n - 1},
&& \Var_\pi[\mu_X |\vec{X}_n] =  \frac{( \beta+ n \bar{X}_n)^2}{(\alpha + n - 1)^2(\alpha + n - 2)}; \label{e:meane2} \\
& \E_{\pi_b}[\mu_X |\vec{X}_n] = \frac{   \bar{nX}_n  }{ n - 1},
&& \Var_{\pi_b}[\mu_X |\vec{X}_n] =  \frac{(  n \bar{X}_n)^2}{( n - 1)^2( n - 2)}. \label{e:vare2}
\end{align}
To be consistent with the baseline choice for $\lambda$, we have retained the choice of $\alpha=0$ and $\beta=0$ for the baseline prior; otherwise one needs to explain why the value of hyper-parameter should depend on the transformation of the parameter. A consequence of this consistency is that the posterior mean does not exist for $n=1$, and posterior variance is infinite
when $n\le 2$, as seen in (\ref{e:vare2}).  But this does not cause trouble because Assumption~\ref{a1} permits a finite number of exceptions as captured
by $k\ge k^*$.  Taking the same $a$ and $d$ as above, we have that the mean function is given by $ad(d-1)^{-1}= a +O(d^{-1})$, and
the variance function is given by $a^2{d^2}[(d-1)^2 (d-2)]^{-1} = a^2d^{-1} + O(d^{-2}).$
Condition  (\ref{eq:very}) is therefore satisfied.  Assumption~\ref{a3}(ii) still follows by the same reasoning, while Condition  (\ref{e:supr}) can be verified from $v(\mu_X) = \mu_X^2$, and $[u'(\mu_X)]^2\sigma_X^2 = \sigma^2_X =\mu_X^{2}.$

As a side note, a reviewer pointed out that the posterior variance in (\ref{e:meane})
is not necessarily decreasing in $n$. This is a known phenomenon because posterior variance depends on actual observations. There is no guarantee that with one more particular observation, which could be a rather ``unlucky" one, we would reduce our uncertainty.  It is for such reasons that we emphasize in the main text the need to have replications to ensure monotonicity, and indeed to meaningfully define ``information".    

 \subsection*{Bernoulli}
Assume that $X_1, \dots, X_n$ are Bernoulli random variables, and hence $\mu_X=p$ and $\sigma^2_X=p(1-p)$.  The conjugate prior on $p$ is the beta distribution with parameters $\alpha$ and $\beta$, $B(\alpha, \beta)$.
By taking $\alpha$ and $\beta$ to zero,  our baseline is $\pi(p)\propto p^{-1}(1-p)^{-1}$. The  posteriors  are beta distributions with means and variances
\begin{align}
& \E_\pi[p | \vec X_n] = \frac{\alpha + n \bar{X}_n}{\beta+ \alpha + n},
&& \Var_\pi[p | \vec X_n] =  \frac{(\alpha + n \bar{X}_n)(\beta + n - n\bar{X}_n)}{(\alpha + \beta +n)^2 (\alpha + \beta +n+1)}; \label{e:meanp} \\
& \E_{\pi_b}[p | \vec X_n] = \bar{X}_n,
&& \Var_{\pi_b}[p | \vec X_n] =  \frac{\bar{X}_n(1- \bar{X}_n)}{ n+1}.\label{e:varp}
\end{align}
As before, (\ref{e:meanp}) implies that we can take $a = \alpha(\alpha + \beta)^{-1}$, $d = \alpha + \beta$, and $T=X$. We then see that (\ref{eq:very}) holds for $u(a) = a$ and $v(a) = a(1-a)$ because $B(\alpha, \beta)$ has mean $a$ and variance $a(1-a)(d+1)^{-1}=v(a)d^{-1}-v(a)[d(d+1)]^{-1}=v(a)d^{-1}+O(d^{-2}). $ Again Assumption~\ref{a3}(ii) follows from the second expression in (\ref{e:varp}), excluding the trivial case where all $X$s are equal.   Condition  (\ref{e:supr}) is  verified because  $[u^{\prime}(\mu_X)]^2\sigma^2_X=p(1-p)\equiv v(\mu_X).$

{As reminded by a reviewer, the baseline prior we used here is not the Jeffreys prior, which is $\pi(p)\propto p^{-1/2}(1-p)^{-1/2}$. If we adopt this prior, the above calculation remains the same, but with an
artificial observation $X_0=0.5$ and changing $n$ to $n+1$. This addition will also make the Assumption~\ref{a3}(ii) always hold because the second expression in \eqref{e:varp} can no longer be zero even if all (observed) $\{X_1, \ldots, X_n\}$ are the same.  Indeed adding an artificial observation of value 0.5 has been a common strategy in dealing with extreme observations, and in rendering much better frequentist properties of the resulting (Bayeian) estimators; see for example \citet{brown2001interval} and the references therein.} 

\subsection*{Poisson}
Assume that $X_1, \dots X_n$ are Poisson random variables, and hence $\mu_{X}=\sigma^2_X=\lambda$.  The conjugate prior on $\lambda$ is the gamma distribution $\Gamma(\alpha, \beta)$.  The prior and the baseline are therefore the same as for the exponential, but the posterior means and variances become
\begin{align}
& \E_\pi[\vec X_n] = \frac{\alpha + n \bar{X}_n}{\beta+ n },
&& \Var_\pi[\lambda | \vec X_n] =  \frac{\alpha + n \bar{X}_n}{(\beta+ n )^2}; \label{e:meano}\\
& \E_{\pi_b}[\lambda | \vec X_n] =  \bar{X}_n,
&& \Var_{\pi_b}[\lambda | \vec X_n] =  \frac{\bar{X}_n}{ n}. \label{e:varo}
\end{align}
The first expression of (\ref{e:meano}) tells us to take $a = \alpha\beta^{-1}$, $d= \beta $ and $T=X$. Condition~(\ref{eq:very}) then holds with $u(a) = a$ and $v(a) = a$ because $\Gamma(\alpha,\beta)$ has mean    $\alpha\beta^{-1}$ and $\alpha\beta^{-2}$. Assumption~\ref{a3}(ii) follows from the second expression in (\ref{e:varo}), excluding the pathological  case where all $X$s are  zero. Condition  (\ref{e:supr}) is verified because $[u^{\prime}(\mu_X)]^2\sigma^2_X=\lambda=v(\mu_x).$
\subsection*{Geometric}
Assume that $X_1, \dots, X_n$ are geometric random variables, and hence $\mu_X=p^{-1}$ and $\sigma^2_x=p^{-2}(1-p)$.   The conjugate prior for $p$ is the beta distribution as given in the Bernoulli example, but with the posterior means and variances given by
\begin{align}
& \E_\pi[p | \vec  X_n] = \frac{\alpha + n }{ \alpha+ \beta + n \bar{X}_n },
&& \Var_\pi[p | \vec X_n] =  \frac{(\alpha + n )(\beta +  n\bar{X}_n - n)}{(\alpha + \beta + n \bar{X}_n)^2 (\alpha + \beta+ n \bar{X}_n+1)}; \label{eq:meang}\\
& \E_{\pi_b}[p |\vec X_n] = \frac{1}{\bar{X}_n},
&& \Var_{\pi_b}[p | \vec X_n] =  \frac{\bar{X}_n - 1}{ \bar{X}_n^2 ( n \bar{X}_n +1)}. \label{eq:varg}
\end{align}
The first expression of (\ref{eq:meang}) then tells us to take  $a = (\alpha+\beta)\alpha^{-1}$ , $d = \alpha$, and $T=X$. Condition~(\ref{eq:very}) then holds with $u(a) = a^{-1}$ and $v(a) = (a-1)a^{-3}$ because $B(\alpha, \beta)$ has mean $a^{-1}$ and variance $a^{-1}(1-a^{-1})(\alpha+\beta+1)^{-1} = (a-1)a^{-3}(d+a^{-1})^{-1}= v(a)d^{-1}-v(a)[ad(ad+1)]^{-1}=v(a)d^{-1}+O(d^{-2})$. The second expression of (\ref{eq:varg}) shows that Assumption~\ref{a3}(ii) is trivially satisfied other than the pathological case where all $X$s are one.
Furthermore, because $[u^{\prime}(\mu_X)]^2\sigma^2_X= [-p^{2}]^{2}p^{-2}(1-p)=p^{2} (1-p)=v(\mu_X)$, condition (\ref{e:supr}) also holds.

\section{Additional Simulations}\label{a:s:D}
In this section we expand upon the initial set of simulations for the normal and exponential settings.  We consider the additional sample sizes of $n=100$ and $500$.  In the normal setting (Figure \ref{f:normkn:extra}) there is also a boundary effect since our estimates of $\Delta^2$ will always be positive (with probability 1) even when the truth is $\Delta = 0$.  To highlight this point a bit more, we also consider the case {where the sample standard deviation is used instead of the known standard deviation} 
(Figure \ref{f:normkn:estsd}) for both methods, as an ad hoc means of countering this boundary effect.  
This means using the sample standard deviation (with the entire sample) in place of $\sigma_0$ when calculating the quick formula \eqref{eq:pd} or in the posterior mean/variance calculations for RPlot.  
As a reminder, the priors in each setting are selected to achieve a particular $r$ (in all cases $r=0.05$) and $\Delta$ (we consider $\Delta = 0,0.5,1,1.25)$, meaning that hyper parameters of the priors change with the sample size.

Starting with Figure \ref{f:normkn:extra} we have the known variance normal setting with $n=100$ (first three rows) and $n=500$ (second three rows).  We can see that the message is basically the same as with $n=1000$.  The more intensive approach (RPlot) based on \eqref{eq:mave} agrees with the faster formula (QuickEst) based approach from \eqref{e:key}.  In this case \eqref{e:key} is exact (though all parameters are still estimated) so it works well as expected.  However, a careful inspection of the first column shows that the methods are slightly biased down when $\Delta = 0$.  At first this could seem odd given that \eqref{eq:key} is exact in this context, however our estimates of $\Delta$ will never be zero, thus there is a meaningful boundary effect biasing the method.  {Moving to Figure \ref{f:normkn:estsd} and focusing on the first column where the boundary effect occurs, we see that this issue is assuaged by using the estimated standard deviation}, 
which provides enough flexibility to absorb the issues at the boundary.  Interestingly, the RPlot and QuickEst methods are also in much stronger agreement (though the variability has increased) 
with the estimated standard deviation as well.

Turning to the exponential setting with a rate parametrization (Figure \ref{f:exp_rate:extra}) and mean parametrization (Figure \ref{f:exp_mean:extra}) we can can explore the differences between our asymptotic approximation and the truth.  In particular, while our plot based approach still works quite well, the QuickEst approach does noticeably worse.  However, this difference gets smaller for larger $n$ especially once compared to Figures \ref{f:exp_rate} and \ref{f:exp_mean}.  The difference between the two parameterizations is not extreme, though the rate based parametrization seems to behave slightly better.

\begin{figure}
    \centering
    \includegraphics[width=1\textwidth]{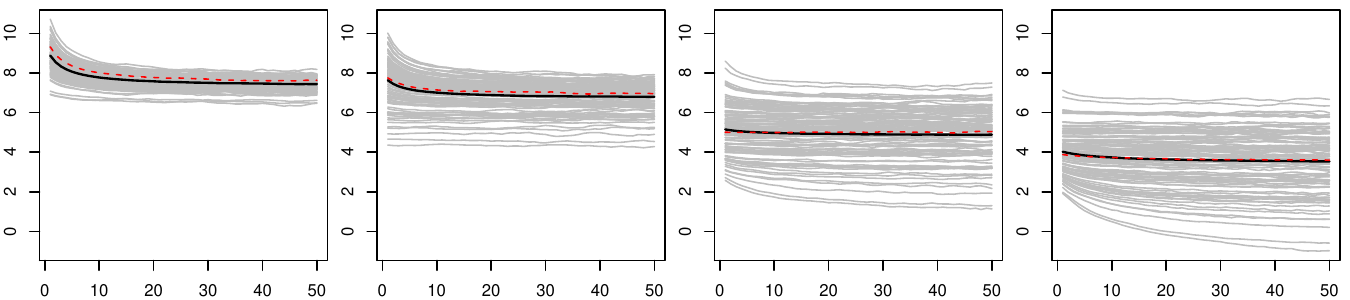}
    \includegraphics[width=1\textwidth]{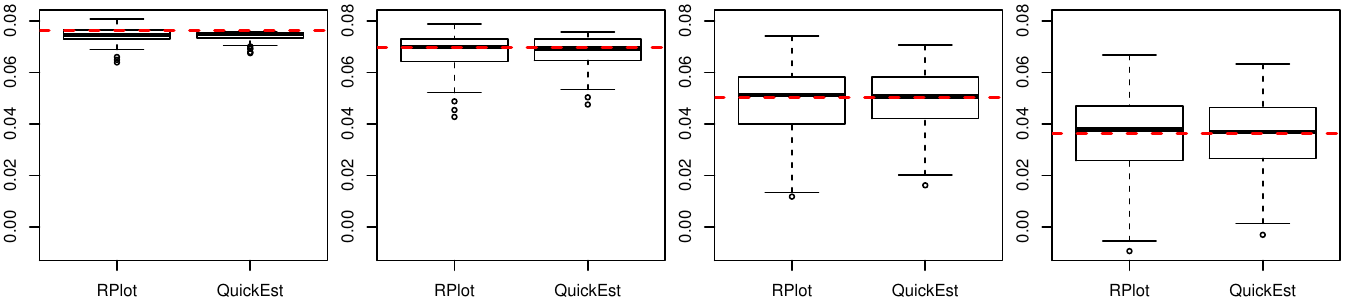}
    \includegraphics[width=1\textwidth]{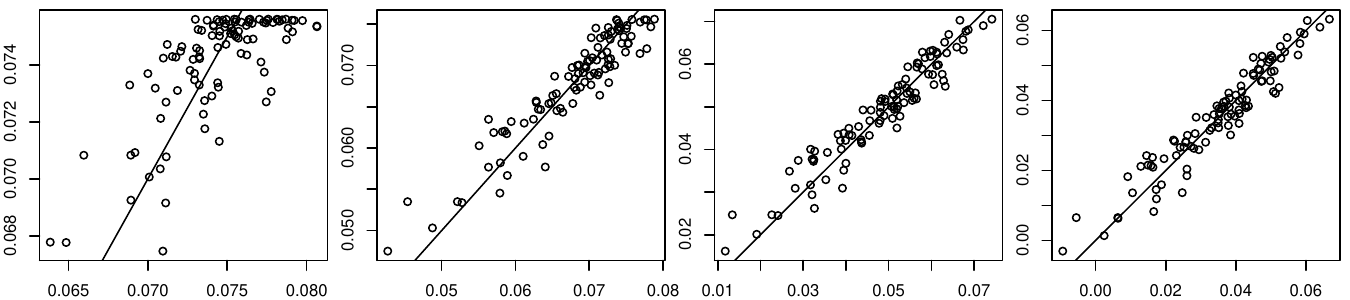}
    \includegraphics[width=1\textwidth]{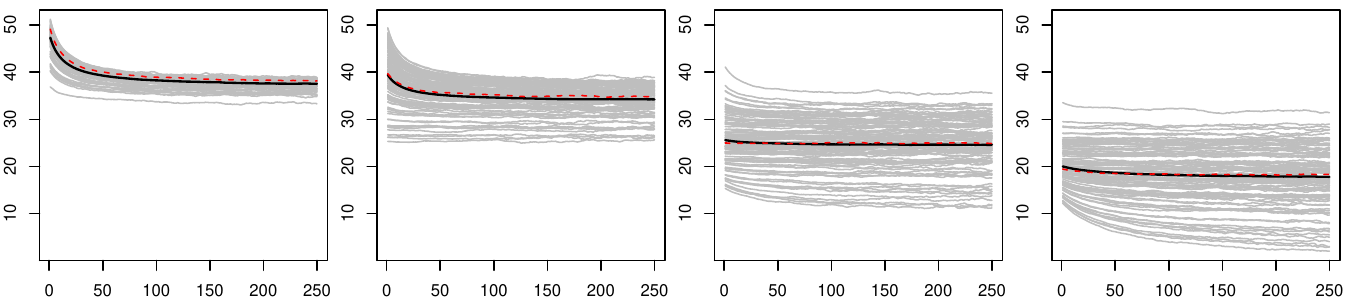}
    \includegraphics[width=1\textwidth]{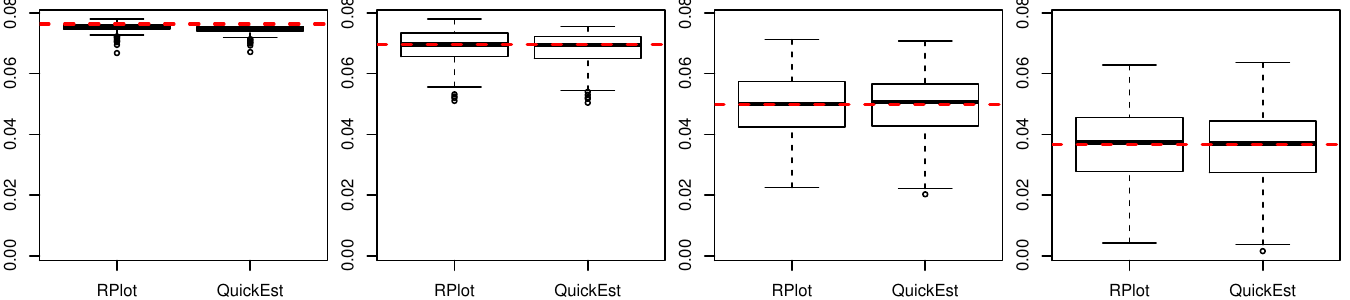}
    \includegraphics[width=1\textwidth]{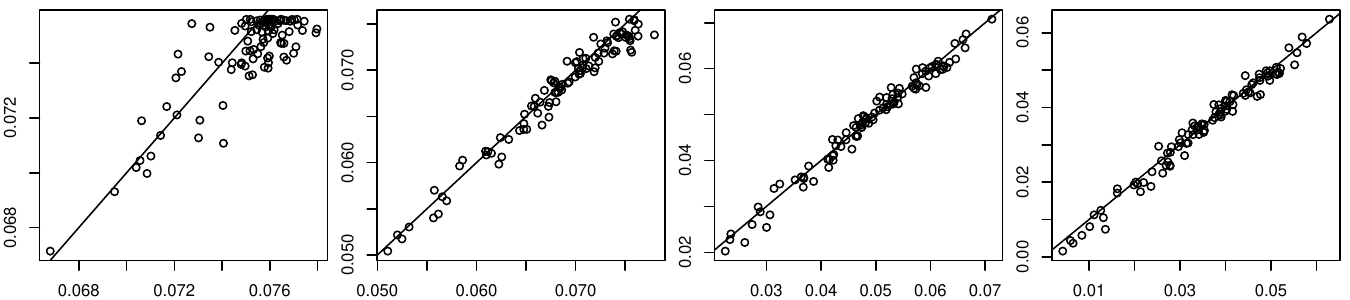}
    \caption{Normal Distribution with $\theta=\mu$, known variance, $n=100$ (first three rows) and $n=500$ (last three rows).  The same caption for Figure~\ref{f:normkn} applies.   }
    \label{f:normkn:extra}
\end{figure}

\begin{figure}
    \centering
    \includegraphics[width=1\textwidth]{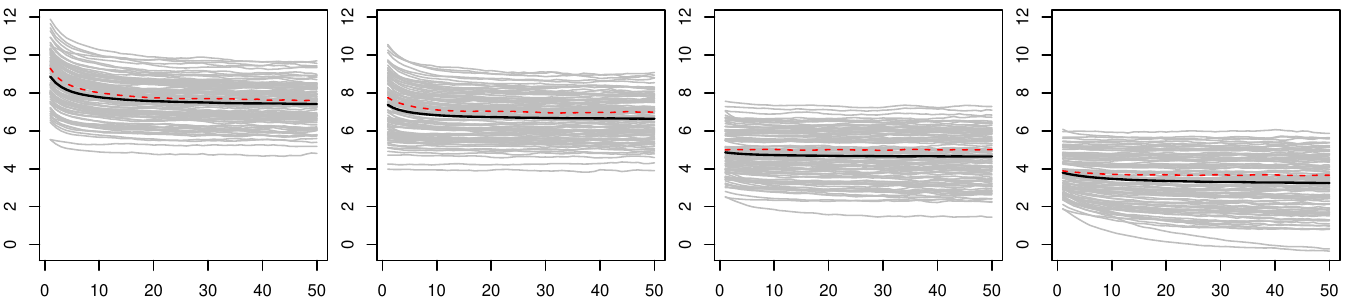}
    \includegraphics[width=1\textwidth]{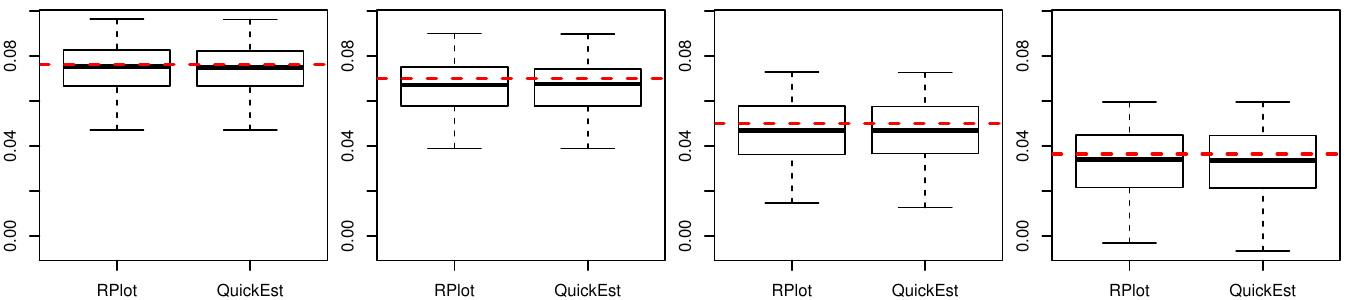}
    \includegraphics[width=1\textwidth]{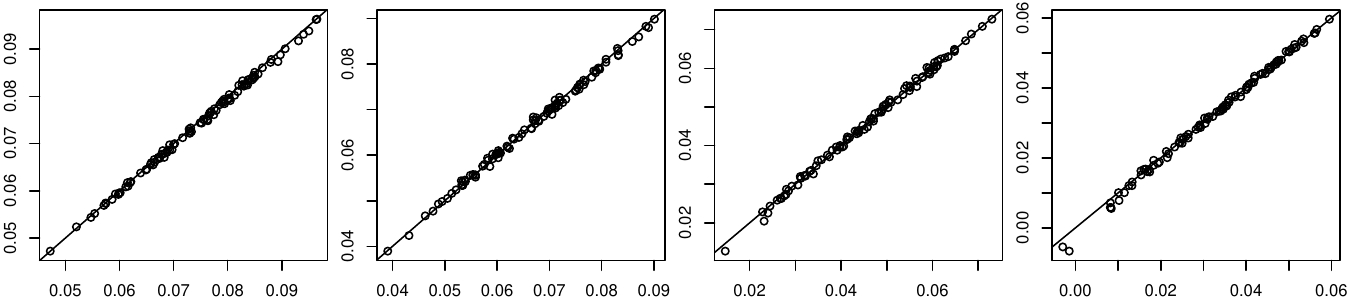}
    \includegraphics[width=1\textwidth]{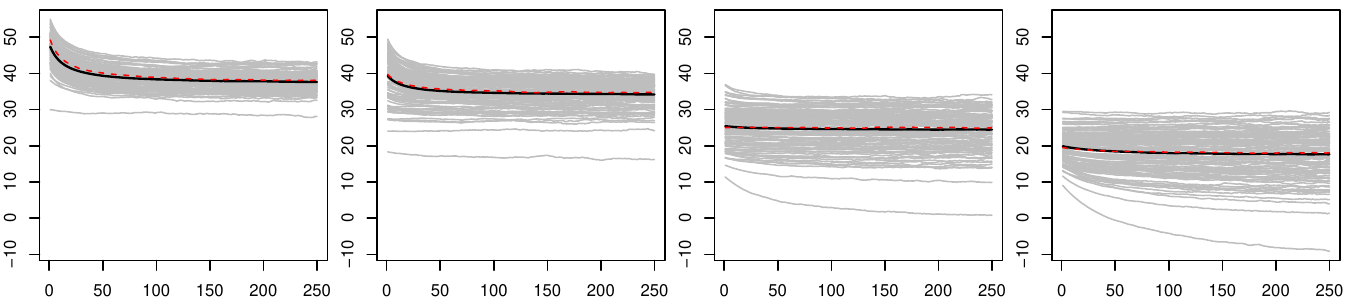}
    \includegraphics[width=1\textwidth]{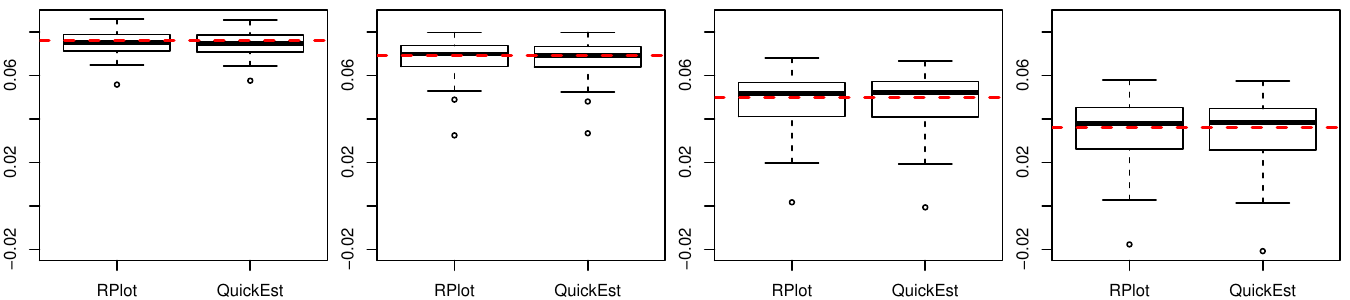}
    \includegraphics[width=1\textwidth]{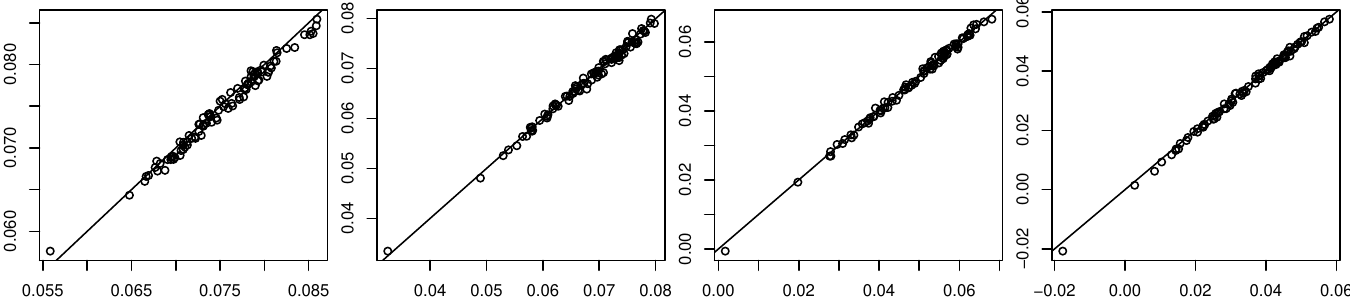}
    \caption{This setting is identical to Figure \ref{f:normkn:extra}, but the quick formula uses the sample standard deviation instead of the true standard deviation.}
    \label{f:normkn:estsd}
\end{figure}

\begin{figure}
    \centering
    \includegraphics[width=1\textwidth]{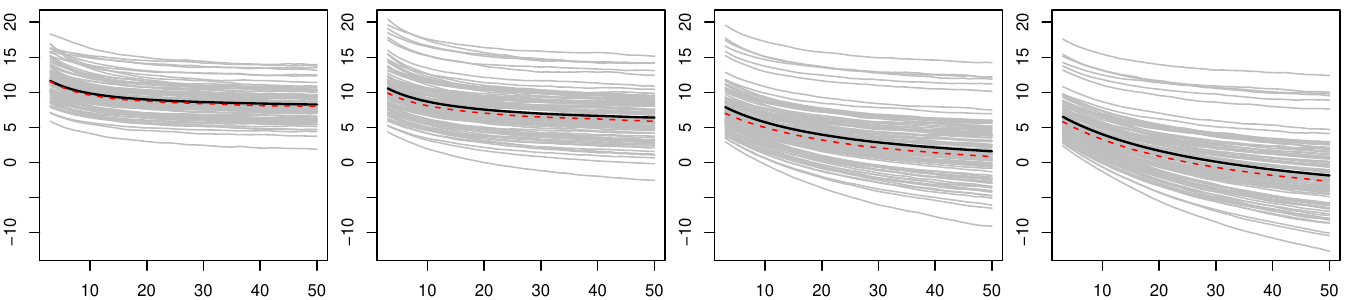}
    \includegraphics[width=1\textwidth]{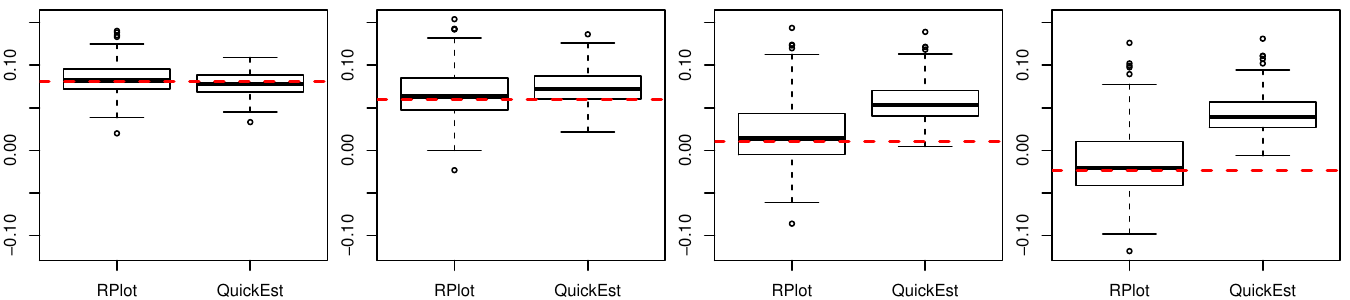}
        \includegraphics[width=1\textwidth]{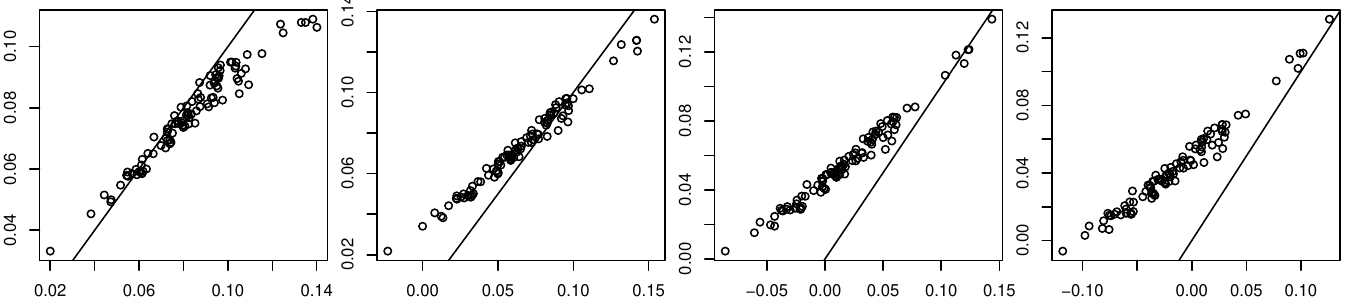}
        \includegraphics[width=1\textwidth]{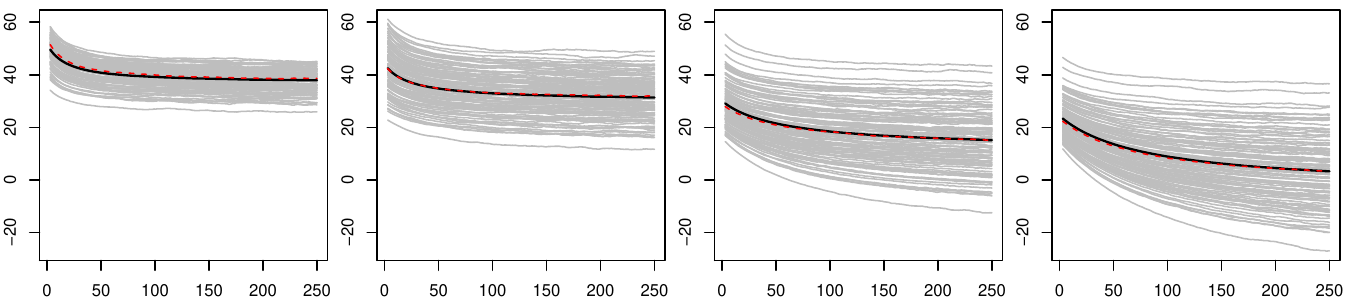}
    \includegraphics[width=1\textwidth]{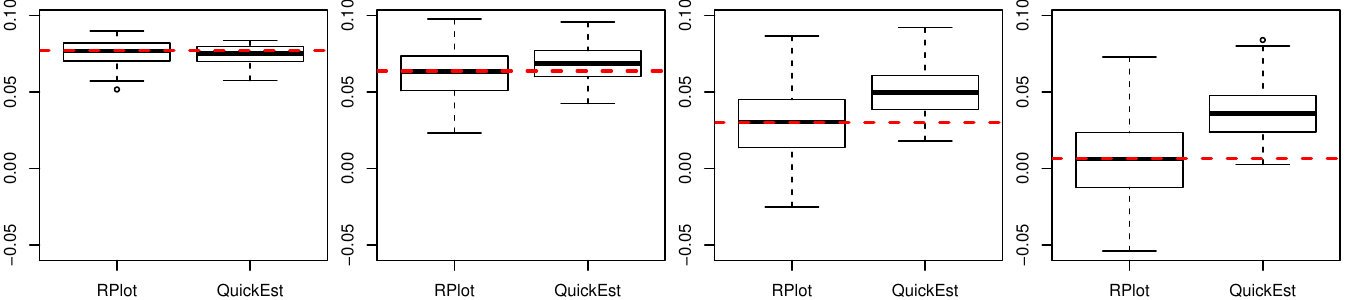}
        \includegraphics[width=1\textwidth]{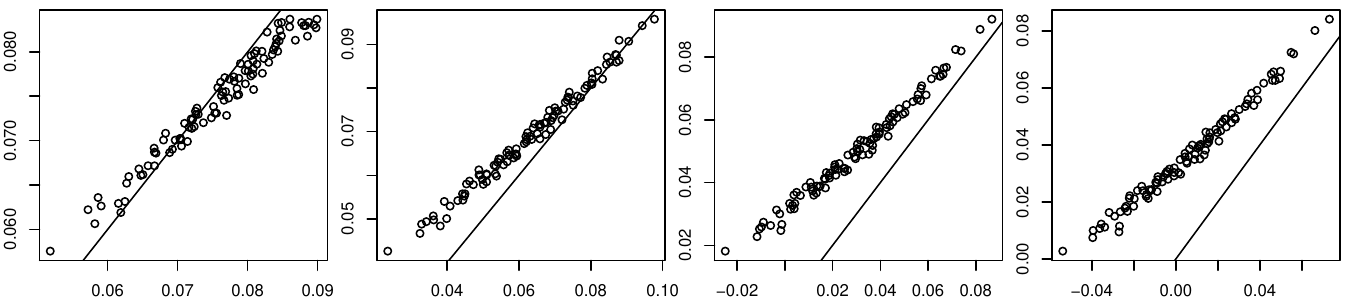}
    \caption{Exponential Distribution with $\theta=1/\mu$, $n=100$ (first three rows) and $n=500$ (last three rows).  The same caption for Figure~\ref{f:normkn} applies. }
    \label{f:exp_rate:extra}
\end{figure}

\begin{figure}
    \centering
    \includegraphics[width=1\textwidth]{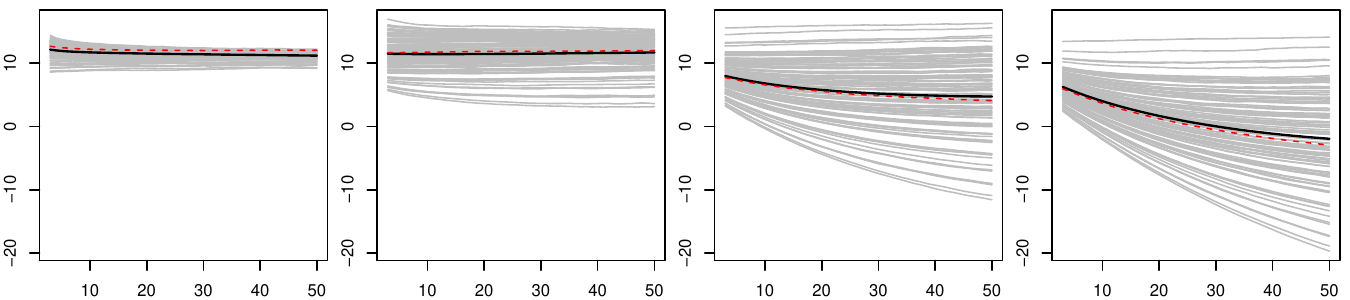}
    \includegraphics[width=1\textwidth]{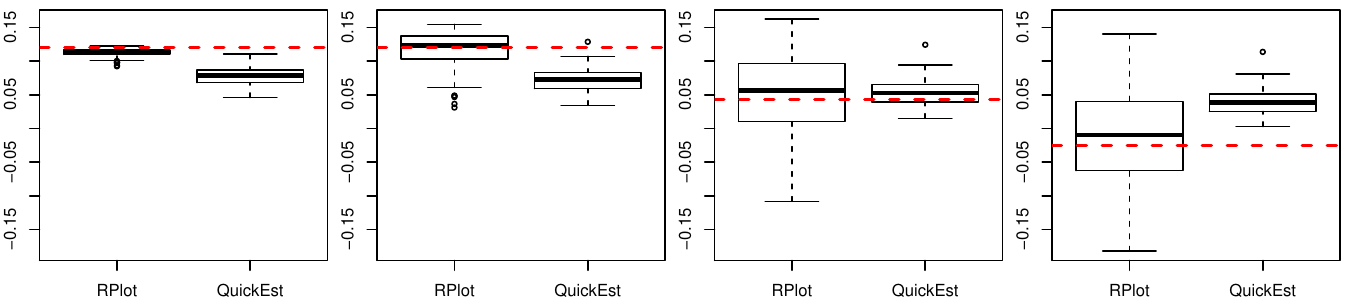}
        \includegraphics[width=1\textwidth]{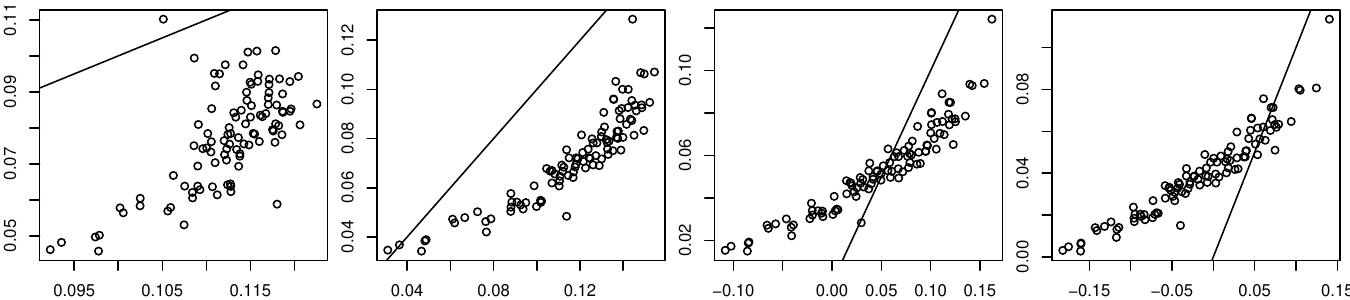}
        \includegraphics[width=1\textwidth]{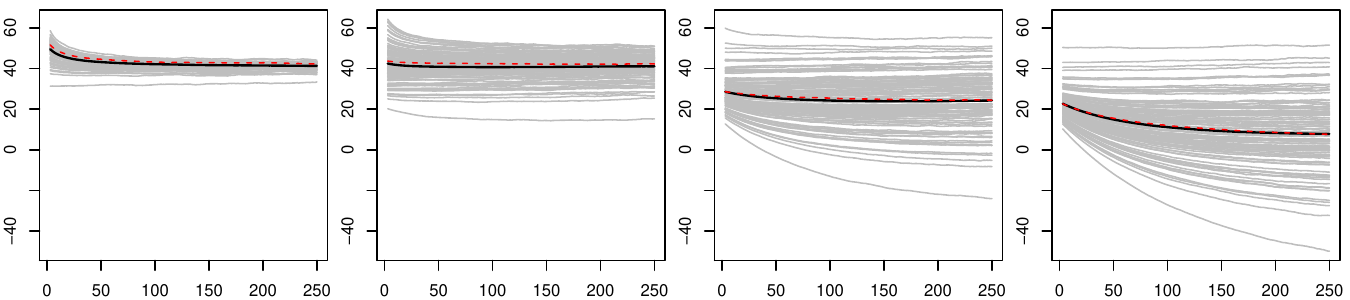}
    \includegraphics[width=1\textwidth]{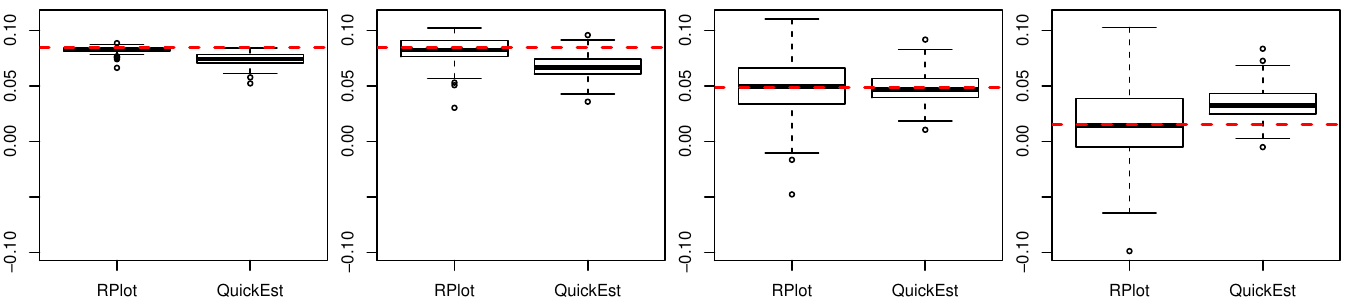}
        \includegraphics[width=1\textwidth]{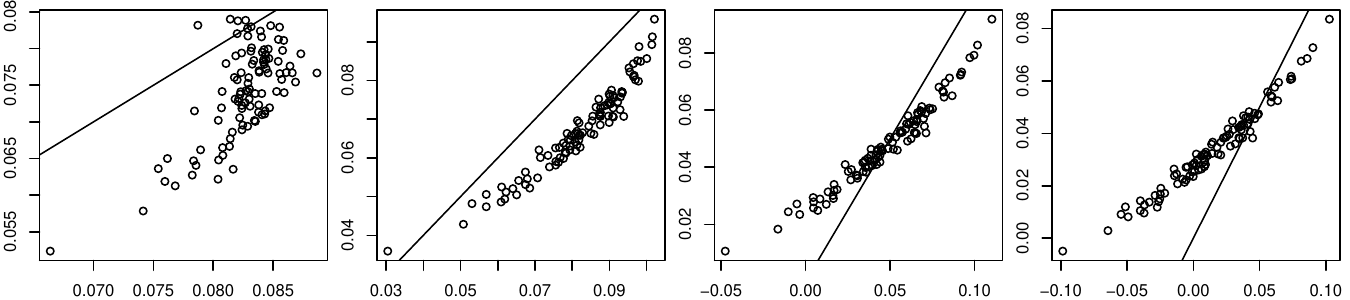}
    \caption{Exponential Distribution with $\theta=\mu$, $n=100$ (first three rows) and $n=500$ (last three rows).  The same caption for Figure~\ref{f:normkn} applies. }
    \label{f:exp_mean:extra}
\end{figure}

\clearpage

\section{Data Application} \label{a:app}
In this section we explore the sensitivity of the results in Section \ref{s:app} to different baselines for the application in Section \ref{s:app}.  In short, we will see that the results summarized in Figure \ref{f:app} are quite robust.  We consider both normal and Cauchy distributions for the baseline on the slope parameter; recall we originally used Cauchy distributions with a scale of 1000 for the baseline and scales of 2.5, 5, and 10 for the prior.  In the first row of Figure \ref{a:f:app} we plot $\hat M_k$ using a Cauchy baseline with scales of 100, 1000, and 10000, respectively, from left to right.  In the second row we use the same setup, but the baseline is changed from a Cauchy to a normal (the priors of interest  are still Cauchy).  In each case, regardless of the choice, the general pattern is the same as the one seen in Figure \ref{s:app}.    There is a small amount of conflict for the smaller scales, but the conflict dissipates quickly for larger sample sizes, resulting in priors that add around 2-6 observations worth of information.  This conclusion is the same regardless of the baseline used and agrees with the results in Section \ref{s:app}.

Lastly, as a test to the sensitivity to the choice of $D$, we swap from using the $L_2$ to $L_1$ distance, i.e. using mean absolute deviation instead of mean squared error.  We do this only for one setting (normal baseline with scale 100) and summarize the results in Figure \ref{a:f:appL1}.  As we can see, compared to Figure \ref{a:f:app}, the results are nearly the same, indicating that the results of Section \ref{s:app} are fairly robust to the choice of $D$ as well.

\begin{figure}[ht]
\begin{center}
\includegraphics[width=.3\textwidth]{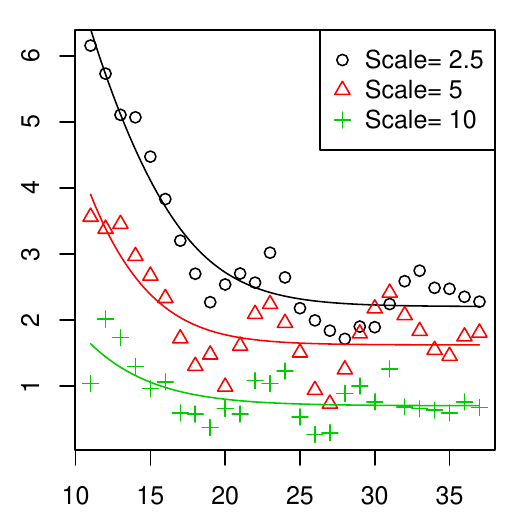}
\includegraphics[width=.3\textwidth]{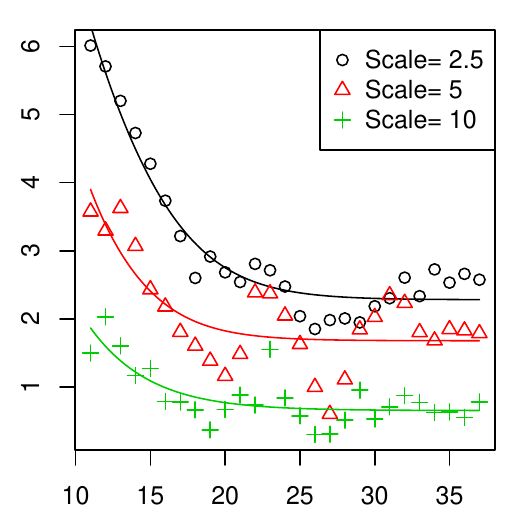}
\includegraphics[width=.3\textwidth]{sup_graphs/App_M_L2_dfb_1_dfp_1_Scale_10000.pdf}
\includegraphics[width=.3\textwidth]{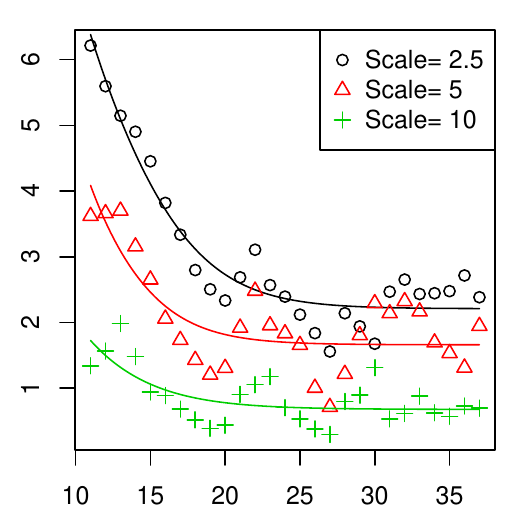}
\includegraphics[width=.3\textwidth]{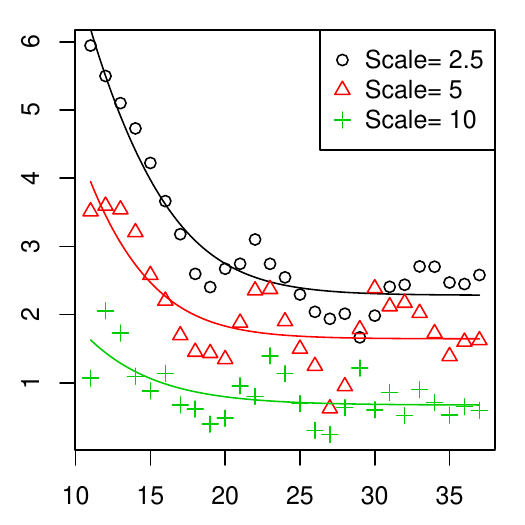}
\includegraphics[width=.3\textwidth]{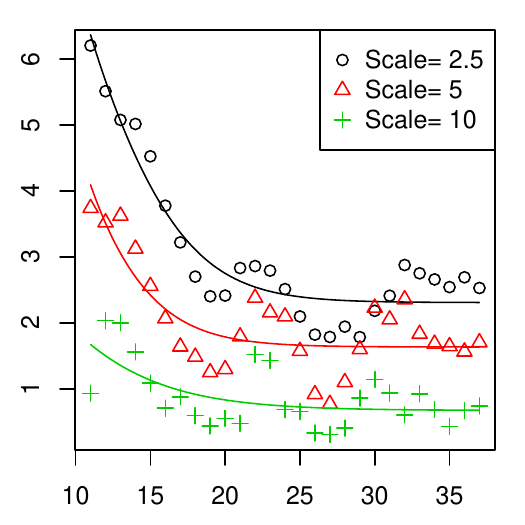}
\vspace{-.2in}
\end{center}
\caption{Recreations of Figure \ref{f:app} using different baseline priors.  The top row uses Cauchy priors with scales of 100, 1000, and 10000, respectively, from left to right.  The bottom row is the same, but the Cauchy family is exchanged for the normal. \label{a:f:app}}
\end{figure}

\begin{figure}[ht]
\begin{center}
\includegraphics[width=.3\textwidth]{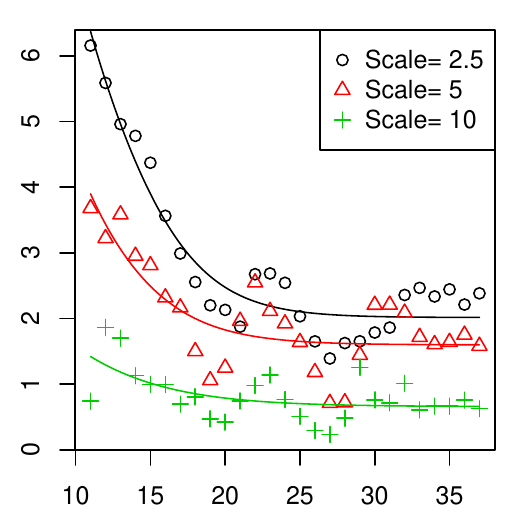}
\vspace{-.2in}
\end{center}
\caption{Recreation of Figure \ref{f:app}, but using MAD instead of MSE for $D$.  The baseline was taken to be normal with scale 100.  \label{a:f:appL1}}
\end{figure}

\end{document}